\documentclass[12pt]{article}
\newcommand{\blind}{0}

\usepackage{setspace}
\usepackage{pifont}% http://ctan.org/pkg/pifont
\newcommand{\cmark}{\ding{51}}%
\newcommand{\xmark}{\ding{55}}%
\usepackage{xcolor}
\usepackage{graphicx}
\usepackage{subcaption}
\usepackage{appendix}
\usepackage{amsmath}
\usepackage{booktabs}
\usepackage[font={small,it}]{caption}
\usepackage{amsmath,amssymb,amsthm,bm}
\usepackage[procnumbered, ruled,vlined]{algorithm2e}
\usepackage{import}
\usepackage{cases}
\usepackage{multirow}
\usepackage{multicol}
\usepackage{float}
\usepackage{mathtools}
\usepackage{diagbox}
\usepackage{natbib}
\usepackage{enumitem}

\usepackage{etoolbox}

\makeatletter
\AtBeginEnvironment{procedure}{\let\c@algocf\c@procedure}
\makeatother

\DeclareMathOperator{\E}{\textsf{E}}
\DeclareMathOperator{\N}{\textsf{N}}
\newcommand{\KLD}[2]{\mathrm{KL} \left( \left. \left. #1 \right|\right| #2 \right) }

\DeclareMathOperator*{\argmin}{arg\,min}

\newtheorem{prop}{Proposition}
\newtheorem*{prop*}{Proposition 1}
\newtheorem*{prop1}{Proposition 1}
\newtheorem*{prop2}{Proposition 2}

\providecommand{\keywords}[1]
{
  \small	
  \textbf{\textit{Keywords---}} #1
}

\begin{document}

\doublespacing

\if0\blind
{
    \title{\bf Bayesian Smoothing and Feature Selection using Variational Automatic Relevance Determination}
  
    \author{Zihe Liu\thanks{Department of Statistics, University of Illinois Urbana-Champaign, Champaign, IL; email: ziheliu2@illinois.edu.} \,,\,
    Diptarka Saha\thanks{Department of Statistics, University of Illinois Urbana-Champaign, Champaign, IL; email: saha12@illinois.edu.} \,\,
    and Feng Liang\thanks{Department of Statistics, University of Illinois Urbana-Champaign, Champaign, IL; email: liangf@illinois.edu.}
    }
    \date{}
} \fi

\if1\blind
{
  \title{\bf Bayesian Smoothing and Feature Selection using Variational Automatic Relevance Determination}
  \date{}
} \fi

\maketitle

\bigskip

\begin{abstract}
This study introduces Variational Automatic Relevance Determination (VARD), a novel approach tailored for fitting sparse additive regression models in high-dimensional settings. VARD distinguishes itself by its ability to independently assess the smoothness of each feature while enabling precise determination of whether a feature's contribution to the response is zero, linear, or nonlinear. Further, an efficient coordinate descent algorithm is introduced to implement VARD. Empirical evaluations on simulated and real-world data underscore VARD's superiority over alternative variable selection methods for additive models.
\end{abstract}
\noindent%
\vfill

\newpage

\section{Introduction}
\label{sec:intro}

We consider the problem of simultaneous smoothing and variable selection for the additive model:
\begin{equation} \label{eq:additive}
f(x_1, \dots, x_p) = f_1(x_1) + \cdots + f_p(x_p),
\end{equation}
where each $f_j(x_j) = \sum_k \beta_{jk} h_{jk} (x_j)$ is a one-dimensional non-parametric function defined on feature $x_j$ through potentially non-linear bases. In the regression setup, the problem boils down to determining which $f_j$ functions are exactly zero and controlling the smoothness of $f_j$ s which are deemed to be non-zero.

% While it might appear similar to a linear regression problem when we represent $f_j$ through a basis expansion , applying traditional variable selection techniques for linear regression isn't straightforward in this context.  This complexity arises because estimating $f_j$ entails more than just determining regression coefficients; it also involves controlling the smoothness (or roughness) of each $f_j$, a crucial aspect closely linked to variable selection.

%The most important factor in the success of additive model is being able to fit a model with proper complexity. This is accomplished from two aspects: 1. smoothness tuning. 2. feature selection.

%%To estimate each $f_j$ while regulating the roughness of each $f_j$, the most naive way is to represent $f_j$ directly with a basis expansion and estimate them with least square estimator. The roughness of each $f_j$ is plainly controlled by each $f_j$'s basis expansion dimension. More basis functions corresponds to less smoothness in the fitted $f_j$. This naive method was effective in low dimensional problems, however the major issue was that roughness represented by the dimensions are not continuous and infeasible to scale up to high dimensional setup for tuning.

The conventional approach to controlling roughness in an additive model involves the use of smoothing splines \citep{wahba1990spline}. 
% This is achieved by having an objective function that penalizes the roughness of each component function $f_j$, quantified as the integral of the square of the second-order derivative of that $f_j$: 
%  \begin{equation}
% \label{eq:smoothing_spline_obj}
% \underset{\{f_j\}_{j=1}^p}{\argmin} \|\mathbf{y} - \sum_{j=1}^pf_j(\bm{x}_j)\|_2^2 + \sum_{j=1}^p\lambda_j \int f_j^{''}(t)^2dt
% \end{equation}
% Here $\mathbf{y} = (y_1, \dots, y_n)^T$ is the response vector, $\{\bm{x}_j = (x_{1j}, \dots, x_{nj})^T\}_{j=1}^p$ are $p$ predictors, and $\{\lambda_j\}_{j=1}^p$ are $p$ hyperparameters to control each $f_j$'s smoothness. 
Upon basis expansion, the objective function of smoothing spline is a form of ridge regression \citep{RIDGE} with each $f_j$'s roughness being controlled by a separate ridge smoothness parameter $\lambda_j$ (see detailed form in Appendix \ref{sec:appendix_regularization_review}). However, dealing with so many parameters becomes impractical for manual tuning as the number of features grows. To address this challenge, a series of works \citep{wood2000, Wood2004, wood2008, wood2010, wood2016} have been developed for learning smoothness parameters adaptively. See \cite{woodbook} for a comprehensive understanding of this line of work. 
% The author also developed an R package \texttt{mgcv} \citep{mgcv_R_Package} which is the state-of-art for learning smoothness parameters efficiently.
However, smoothing spline methods face limitations in tasks concerning feature selection due to their Ridge-like objective function.

On the other hand, a popular method to facilitate feature selection in additive models is through the application of a \textit{group Lasso} type objective function \citep{yuan2006model}, which achieves exact sparsity. Prominent methods in this field include COSSO \citep{lin2006}, SPAM \citep{ravikumar2007spam, ravikumar2009}, and GAMSEL \citep{alex2015}. However, these approaches are not ideal for smoothing tasks. This is because they typically employ a single universal roughness penalty parameter for different \( f_j \), limiting their flexibility to accommodate varying levels of smoothness. While it is theoretically feasible to assign distinct penalty parameters to each \( f_j \) to control their smoothness individually, the computational demands of such an approach can become prohibitive. Further exploration of these issues is provided in Section \ref{sec:background}.

The Bayesian method often used to induce sparsity in additive models is the spike and slab approach \citep{he2022bayesian, scheipl2012}. In this method, an auxiliary indicator variable \( Z \) determines whether the slab or spike component dominates, thereby influencing whether the parameter is included in the final model or not. While this strategy potentially enables feature selection by setting a cutoff on the posterior inclusion probability \( P(Z = 1 | \mathcal{D}) \), it does not naturally achieve exact sparsity. This limitation has several drawbacks, as discussed in Section \ref{sec:bayesian methods}. In this paper, our focus is on strategies that achieve exact sparsity.

% \textcolor{red}{1. smoothing parameter (slab prior variance) pre-fixed, use for all features. 2. hard to tune, too many hyperparameter}

To summarize, in the literature, we have methods that perform well in either of individual smoothing (e.g. smoothing spline) or feature selection (e.g. COSSO, SPAM, GAMSEL) tasks, but there exists a gap in the literature for a method that can perform both tasks well simultaneously. 

To address this gap, we propose a novel, fully Bayesian, and computationally efficient coordinate descent algorithm that integrates concepts from variational Bayes and Automatic Relevance Detection \citep{mackay1995probable, neal1996bayesian}. This approach allows us to simultaneously manage smoothing and feature selection tasks. Inspired by Automatic Relevance Determination (ARD), a key concept in empirical Bayes methods for sparsity modeling, our method learns smoothness by estimating individual prior variances for coefficients after basis expansion. This method naturally achieves exact sparsity when prior variances are estimated as zero \citep{tipping2003fast, wipf2007new}. Further, similar to Lasso, our algorithm involves just one hyperparameter, which can be easily tuned using cross-validation. Through several experiments on real and synthetic datasets, we demonstrate the superiority of our algorithm over current state-of-the-art methods.

The rest of this paper is organized as follows. In Section 2, we review background knowledge about the additive model to motivate our method. In Section 3, we introduce our framework with automatic relevance determination. In Section 4, we present our coordinate descent algorithm along with other implementation details. In Section 5, we compare our method's performance with other sparsity-inducing methods such as SPAM and GAMSEL, on real and simulated datasets.

\section{Background}
\label{sec:background}
\subsection{Regularization Methods}
\label{sec:regularization_methods}
Consider a dataset with $n$ observations, each consisting of a response $y_i$, and  $p$ predictors $\{{x}_{ij}\}_{j=1}^p$. Regularization in additive models \eqref{eq:additive} is done through minimizing an objective that combines residual sum of squares, %\text{RSS}(f) = \|\mathbf{y} - \sum_{j=1}^p f_j(\bm{x}_j)\|_2^2$ 
$\text{RSS}(f) = \sum_{i=1}^n \big [y_i - \sum_{j=1}^p f_j(x_{ij}) \big]^2$
with a penalty term $Penalty(f)$. For our purposes here, the penalty may be expressed abstractly as follows:
% The goal of regularization methods in the additive model is to find a function $f$ of the additive form \eqref{eq:additive} that minimizes the following objective:

% \begin{equation}
% \label{eq:regularization_obj}
% \text{Objective}(f) = \text{RSS}(f) + \text{Penalty}(f)
% \end{equation}

% Here, $\text{RSS}(f)$ represents the residual sum of squares, given by $\|\mathbf{y} - \sum_{j=1}^p f_j(\bm{x}_j)\|_2^2$. We use the notation $f_j(\bm{x}_j)$ to represent the vector $(f_j(x_{1k}), \cdots, f_j(x_{nj}))^T$ for clarity.

% At a conceptual level, the formulation of the penalty term, denoted above as $Penalty(f)$, in the context of the four regularization techniques mentioned above can be expressed as follows:

\begin{equation}
\label{eq:pen_form}
    Penalty(f) =  \underbrace{\sum_{j=1}^p \lambda_jJ_R(f_j)}_\text{Ridge} + \underbrace{\lambda\sum_{j=1}^pJ_L(f_j)}_\text{Group Lasso}.
\end{equation}

% Within this formulation, the various components can be understood as follows:

\begin{itemize}
    \item $\sum_{j=1}^p \lambda_j J_R (f_j)$ represents the `ridge' component of the penalty. The function $J_R(f_j)$ quantifies the roughness of each $f_j$ using a quadratic penalty structure. The parameters $\lambda_1, \ldots, \lambda_p$ are non-negative smoothness parameters that independently control the degree of roughness for each $f_j$.

    \item $\lambda\sum_{j=1}^pJ_L(f_j)$ corresponds to the `group Lasso (gLasso)' part of the penalty, encouraging sparsity and sometimes having a global smoothing effect. The function $J_L(f_j)$ imposes a first absolute moment penalty.
    The parameter $\lambda$ is a non-negative tuning parameter which induces and controls sparsity. 
    % When $J_L(f_j)$ also captures roughness, the group lasso parameter $\lambda$ has a global smoothing effect on all $f_j$.
\end{itemize}

Notice that in the penalty structure \eqref{eq:pen_form}, the ridge part has individual hyperparameters $\lambda_1,\cdots, \lambda_p$, while the hyperparameter for the gLasso part is shared. To our knowledge, no existing work assigns individual hyperparameters for Lasso to control smoothness as with ridge. The reason that it's feasible to introduce $p$ hyperparameters for ridge is that the optimal coefficient solution for ridge is in closed form in terms of these $p$ hyperparameters, and subsequently, the cross-validation-based objective can also be represented in closed form in terms of these $p$ hyperparameters (see, e.g., \cite{woodbook}). However, a similar method does not apply to Lasso and its variants such as gLasso if one wants to use $p$ individual hyperparameters, simply because the Lasso solution is not in closed form, making it very challenging to optimize $p$ individual hyperparameters simultaneously.

In the penalty structure \eqref{eq:pen_form}, one has the flexibility to set $\lambda$ and $\lambda_1, \ldots, \lambda_p$ to zero, resulting in exclusively ridge or gLasso regularization. The regularization techniques in this discussion, namely Smoothing splines, COSSO, SPAM and GAMSEL,  are summarized within this framework in Table \ref{tb:regularization_methods}. For detailed expressions of $J_R$ or $J_L$ for each method, please refer to the appendix \ref{sec:appendix_regularization_review}.

\begin{table}[ht]%[ht]
\scriptsize
\begin{center}
\caption{Regularization Methods in Additive Model}
\label{tb:regularization_methods}
\begin{tabular}{ |c|c||c|c||c|c||c|c| }
\hline
Method & Type & $J_R$ & $\lambda_1, \cdots, \lambda_p$ & $J_L$ & $\lambda$ & Smoothing & Sparsity\\
\hline
\texttt{mgcv} & Ridge &Roughness & Adaptive & \xmark &  \xmark & Individual & \xmark\\
COSSO & gLasso &\xmark & \xmark & Magnitude \& Roughness & Tuning & Global & \cmark\\
SPAM & gLasso & \xmark & \xmark & Magnitude & Tuning & \xmark & \cmark\\
GAMSEL & Mixture & Roughness & Pre-specified & Magnitude \& Roughness & Tuning & Global & \cmark\\
\hline
\end{tabular}
\end{center}
\end{table}

Smoothing Spline is ridge-only, excelling in individual smoothing but not inducing exact sparsity. COSSO and SPAM belong to the gLasso category. COSSO combines magnitude and roughness for sparsity and global smoothing, while SPAM focuses solely on sparsity. In GAMSEL, the Lasso parameter $\lambda$ serves as a global roughness measure and induces sparsity, thus combining ridge and gLasso -- albeit in a pre-specified manner. We must emphasize that the global smoothing arising in COSSO and GAMSEL is an unintended consequence of the shared group lasso penalty $\lambda$ and can be problematic when handling functions with massively varying roughness. In summary, none of these methods effectively balance both individual smoothing and feature selection tasks.

Indeed, we argue via the following logical chain that it's not just the listed methods but any regularization approach with penalty structure such as \eqref{eq:pen_form} will face challenges in simultaneously achieving both smoothing and selection. To enable both individual smoothing and selection, the penalty must encompass both ridge and gLasso components. Once both components are included in the penalty, the adaptative learning of individual smoothness parameters becomes impractical due to the multitude of hyperparameters to tune. Further, the global impact of the gLasso parameter $\lambda$ complicates the task of achieving individual smoothing. In essence, this regularization framework presents inherent challenges when attempting to simultaneously balance the objectives of smoothing and variable selection.

% \begin{enumerate}
%     \item 
% To enable both individual smoothing and selection, the penalty must encompass both ridge and group lasso components; otherwise, it lacks the capacity for either individual smoothing or inducing sparsity.

% \item Once both ridge and group lasso components are included in the penalty, the adaptative learning of individual smoothness parameters ($\lambda_1, \cdots, \lambda_p$) becomes impractical due to the multitude of hyperparameters to tune.

% \item  Further, the group lasso part also influences smoothness. The global impact of the group lasso parameter $\lambda$ further complicates the task of achieving individual smoothing.
% \end{enumerate}

\subsection{Bayesian Methods}
\label{sec:bayesian methods}

The idea of priors in Bayesian methods is inherently linked with frequentist regularization techniques. In our setup,  applying a penalty in smoothing splines is equivalent to assigning each basis coefficient $\mathbf{\beta}_j$ an independent Gaussian prior 
$ p(\bm{\beta}) =  \prod_{j=1}^p N(\bm{\beta}_j|\mathbf{0}, \frac{\sigma^2}{\lambda_j}S_j^{-1})$
Here,  $\{S_j, \lambda_j\}_{j=1}^p$, denote smoothing matrices and smoothness parameters for each $f_j$ respectively, and $\sigma^2$ is the error variance. The resulting posteriors can then be used to find data-adaptive estimates of the smoothness parameters $\{\lambda_1, \cdots, \lambda_p\}$ \citep{woodbook}. Moreover, Bayesian strategies enable uncertainty quantification and thus adds interpretability.

For feature selection in additive models, authors \citep{scheipl2012, he2022bayesian} have often relied on the spike and slab prior $\label{eq: spike_slab_prior}
    p(\bm{\beta}) =  \Pi_{j=1}^p \{(1-\pi)p_{spike}(\bm{\beta}_j) + \pi p_{slab}(\bm{\beta}_j)\}.$
% the most popular Bayesian approach involves the use of a on the coefficients :
%     \begin{equation}
%     \label{eq: spike_slab_prior}
%     p(\bm{\beta}) =  \Pi_{j=1}^p \{(1-\pi)p_{spike}(\bm{\beta}_j) + \pi p_{slab}(\bm{\beta}_j)\}.
%     \end{equation}
 % Here, the prior $p(\bm{\beta})$ is formulated as a mixture of two components: a `spike' distribution, which is highly concentrated around zero, and a `slab' distribution. The parameter $\pi \in (0, 1) $ represents the prior inclusion probability. Feature selection is achieved by setting a cutoff value on each feature's posterior inclusion probability.
The more dominant the slab part is the more likely it is that the parameter is active in the final model and vice versa. Thus, measuring the posterior inclusion probability of the spike part for a feature provides a principled strategy for feature selection. However, there are several drawbacks to this approach. Firstly for our purposes, it's not clear how individual smoothing is considered within the model. Additionally, determining the numerous prior hyperparameters (spike variance, slab variance,  prior inclusion probability) can be challenging in practice. Finally, this approach doesn't induce exact sparsity adding subjective biases in feature selection.
 
 % Computationally, the posterior can be learned using methods such as Markov Chain Monte Carlo (MCMC) or approximated with Variational methods.

% However, there are drawbacks to the spike and slab approach. Unlike regularization methods, it's not clear how individual smoothing is considered within the model. Additionally, determining the numerous prior hype rparameters (for the spike distribution, slab distribution, and prior inclusion probability) can be challenging in practice. Moreover, unlike regularization methods that naturally induce sparsity by yielding exact zero solutions, spike and slab method requires an additional step of examining posterior inclusion probabilities, which can be subjective and less straightforward for model interpretation.

\subsection{Automatic Relevance Determination}
Developed by \cite{mackay1995probable} and \cite{neal1996bayesian}, Automatic Relevance Determination (ARD) initially emerged in the context of neural networks for network compression. ARD induces sparsity in a parametric model with parameters $\Theta = (\theta_1, \cdots, \theta_p)$ by assigning each parameter $\theta_j$ an independent normal prior with a zero mean and a learnable variance $r_j^2$: $p(\Theta) = \Pi_{j=1}^pN(\theta_j|0, r_j^2)$. This approach aligns well with our goals because the roughness of each $f_j$ is encapsulated in its prior variance. Therefore, as discussed in the previous section on smoothing splines, individual smoothing reduces to finding appropriate prior variance parameters from a Bayesian perspective. A small prior variance suggests less relevance compared to input features with larger prior variances. Recent advancements in ARD have shown that these prior variances can not only be small but can also reach exact zero \citep{tipping2003fast, wipf2007new}, indicating a point mass prior (and thus posterior) distribution at zero for $\theta_j$, and thus inducing exact sparsity. In the context of additive models, our approach extends the group-based version of ARD, offering a seamless way to achieve both smoothing and selection by learning each component's prior variance. This adaptive learning of smoothness for every component, coupled with the potential for exact sparsity when the prior variance reaches zero, makes ARD an ideal fit for our objectives, essentially serving as a \textit{`one-stone-two-birds'} method.

\section{VARD for GAM}
Our framework is developed by combining with the ideas of variational inference and ARD. Hence, we term our method as \textit{Variational Automatic Relevance Determination (VARD)}. VARD excels at simultaneously achieving exact sparsity while fitting the smoothness of individual components $f_j$ adaptively. As a bonus feature, we can also distinguish if each $f_j$ is a zero, linear, or nonlinear function.

\subsection{Setup}
\label{sec:setup}
Given a data set of $n$ observation with centered response $\mathbf{y} = (y_1, \dots, y_n)^T$ and $p$ centered predictors $\{\mathbf{x}_j = (x_{1j}\dots x_{nj})^T\}_{j=1}^p$. Consider the following additive model,
\begin{equation}
\label{eq:additive_model_regression}
    \mathbf{y} = \sum_{j=1}^p f_j(\bm{x}_j) + \bm{\epsilon}
\end{equation}
where $\bm{\epsilon} = (\epsilon_1, \cdots, \epsilon_n)^T$ is an i.i.d. noise vector following Gaussian distribution with mean zero and variance $\sigma^2$.Each $f_j$ is represented using a basis expansion consisting of one linear basis term and 
$d_j$ nonlinear basis terms, as expressed by:
    \begin{equation}
    \label{eq:basis_expansion}
        f_j(x) = \beta_{j0}x + \sum_{k=1}^{d_j}\beta_{jk}h_{jk}(x),
    \end{equation}
where $\{h_{jk}\}_{k=1}^{d_j}$ are $d_j$ nonlinear basis, $\beta_{j0}$ is the coefficient for linear basis term and $\{\beta_{jk}\}_{k=1}^{d_j}$ are the $d_j$ coefficients of nonlinear basis terms associated with  $f_j$.  

Without loss of generality, consider basis expansion with the following properties: 
    \begin{enumerate}
        \item The $j$-th feature's nonlinear basis matrix $H_j := [h_{jk}(x_{ij})]_{i=1, k=1}^{n, d_j}$ is centered and is orthogonal to itself and to the linear component \(\bm{x}_j\) of the \(j\)-th feature. This alleviates potential ambiguities in distinguishing between 
$f_j$ as a linear or nonlinear function when it is non-zero.
        \item The smoothing matrix \(S_j := [\int h^{''}_{jk_1}(x)h^{''}_{jk_2}(x)\mathrm{d}x]_{k_1=1,k_2=1}^{d_j,d_j}\) for the nonlinear basis functions of the \(j\)-th feature is set to the identity matrix \(I_{d_j}\).
    \end{enumerate}    
Note, these orthogonality constraints apply exclusively within the context of each individual feature. For distinct feature indices \(j \neq j'\), the inner products \(\bm{x}_j^T\bm{x}_{j'}\), \(\bm{x}_j^TH_{j'}\), and \(H_j^T H_{j'}\) may assume arbitrary values.  While they may seem non-trivial at first glance, such choices of nonlinear basis not only exist but can also be constructed by standardizing nonlinear basis terms that initially don't satisfy these criteria without altering the functional space they represent. Further details regarding the standardization procedure can be found in Appendix \ref{sec:appendix_standardization}.
% \begin{remark}
% We center $\mathbf{y}$, $\{\bm{x}_j\}_{j=1}^p$ and $\{H_j\}_{j=1}^p$ for the purpose of ignoring intercept. By requiring that each feature's linear component and nonlinear part component be orthogonal to each other, we alleviate potential ambiguities in distinguishing between 
% $f_j$ as a linear or nonlinear function when it is non-zero. To further streamline our framework, we enforce orthogonality within the nonlinear component itself and ensure that the smoothing matrix is set to the identity matrix.  It is essential to emphasize that these orthogonality constraints apply exclusively within the context of each individual feature. For distinct feature indices \(j \neq j'\), the inner products \(\bm{x}_j^T\bm{x}_{j'}\), \(\bm{x}_j^TH_{j'}\), and \(H_j^T H_{j'}\) may assume arbitrary values.
% \end{remark}

% It is straightforward to have basis expansion with the first property where the linear part and nonlinear part are separated as in \eqref{eq:basis_expansion} (e.g. see \cite{hastie2009elements}). As for the second property, (a) is trivial: if any nonlinear basis matrix $H_j$ is not centered, a simple remedy is to shift it by the mean value. While (b) and (c) may seem non-trivial at first glance, such choices of nonlinear basis not only exist but can also be constructed by standardizing nonlinear basis terms that initially don't satisfy (b) and (c), without altering the functional space they represent. Further details regarding the standardization procedure can be found in appendix \ref{sec:appendix_standardization}.
After this basis expansion, we have the matrix representation of model \eqref{eq:additive_model_regression} as the following sum of $2p$ many terms:
\begin{equation}
\label{eq:additive_model_regression_matrix}
    \mathbf{y} = \sum_{j=1}^{2p}Z_j\bm{\beta}_j + \mathbf{\epsilon}.
\end{equation}
Here the non-linear components $\{Z_j\}_1^p:=\{H_j\}_1^p$, $\{\bm{\beta}_j\}_1^p:= \{(\beta_{j1}, \cdots, \beta_{jd_j})^T\}_{j=1}^p$ are confined in the first $p$ terms and the linear components $\{Z_j\}_{p+1}^{2p} := \{\bm{x}_j\}_1^p$, $\{\bm{\beta}_j\}_{p+1}^{2p} := \{\beta_{j0}\}_{j=1}^p$ in the final $p$ terms. By inducing component-wise sparsity in equation \eqref{eq:additive_model_regression_matrix}, we can determine whether each $f_j$ is zero, linear, or nonlinear according to the criteria laid down in Table \eqref{tab:feature class based on betaj}.

\begin{table}[h]
\centering
\caption{Feature Classification based on coefficients of basis terms}
\label{tab:feature class based on betaj}
\begin{tabular}{l|p{10cm}}
\toprule
\textbf{Category} & \textbf{Condition} \\
\midrule
Zero & All coefficients ($\beta_{j0}, \beta_{j1}, \cdots, \beta_{jd_j}$) are zero. \\
\midrule
Linear & $\beta_{j0}$ is nonzero; $\beta_{j1}, \cdots, \beta_{jd_j}$ are zero. \\
\midrule
Nonlinear & At least one $\beta_{j1}, \cdots, \beta_{jd_j}$ is nonzero. \\
\bottomrule
\end{tabular}
\end{table}

In all our future discussions, we will define the dimensionality of $Z_j$ for linear terms by setting $d_j = 1$ for all $j \in {p + 1, \ldots, 2p}$. Finally, we denote the Gram matrix of each component $Z_j$ as $V_j := Z_j^T Z_j$. Due to our orthogonality assumptions, all Gram matrices are diagonal. For simplicity, we will refer to the $k$-th diagonal element of $V_j$ as $v_{jk}$.

% \fl{do we need this ?}
% \begin{equation}
% \label{eq:gram_matrix_diag}
% v_{jk} := \text{the } k\text{-th diagonal element of the diagonal Gram matrix } V_j, \quad j \in {1, \cdots, 2p}
% \end{equation}

\subsection{Group ARD Prior}

In our framework, we start by extending the idea of ARD to group selection on model \eqref{eq:additive_model_regression_matrix} by assigning coefficients $\bm{\beta} =(\bm{\beta}_1, \cdots, \bm{\beta}_{2p})$ with the following prior:
\begin{equation}
\label{eq:prior}
    p(\bm{\beta}) = \Pi_{j=1}^{2p}p_j(\bm{\beta}_j) =  \Pi_{j=1}^{2p}\N(\bm{\beta}_j|\mathbf{0}, r_j^2I_{d_j})
\end{equation}
where $\{r_j^2\}_{j=1}^{2p}$ are the $2p$ non-negative prior variances for each component that we will learn. 

As discussed in the previous section, the $r_j^2$ values that are learned to be zero determine which $\beta_j$ coefficients will be zero. This, in turn, dictates whether a feature's contribution is zero, linear, or non-linear. 

Further, $r_j^2$ can also act as smoothness parameter. This can be seen by viewing our model setup \eqref{eq:additive_model_regression_matrix} with prior \eqref{eq:prior} from the point of view of the classic smoothing spline. Notice that the negative log-likelihood of the posterior distribution of $\bm{\beta}$ of our model has the following form:
\begin{equation}
\label{eq:apriori}
-\log p(\bm{\beta} | \mathbf{y}) \propto \|\mathbf{y} - \sum_{j=1}^{2p}Z_j\bm{\beta}_j\|_2^2 + \sum_{j=1}^{p}\frac{\sigma^2}{r_j^2}\bm{\beta}_j^T S_j\bm{\beta}_j + \sum_{j=p+1}^{2p}\frac{\sigma^2}{r_j^2} \bm{\beta}_j^2.
\end{equation}
Here the smoothing matrices $S_j$ for nonlinear components are recovered in equation \eqref{eq:apriori} simply because our basis are standardized such that $\{S_j = I_{d_j}\}_{j=1}^p$. The first two terms on the RHS of \eqref{eq:apriori} are equal to the classic smoothing spline objective after basis expansion: 
$$\|\mathbf{y} - \sum_{j=1}^pf_j(\bm{x}_j)\|_2^2 + \sum_{j=1}^p\frac{\sigma^2}{r_j^2} \int f_j^{''}(t)^2dt. $$
This implies that our model setup is equivalent to the classic smoothing spline except for additional ridge-like penalty terms imposed on linear components, and our prior variances $\{r_j^2\}_{j=1}^p$ for the nonlinear components play the role as the smoothness parameter here.

\subsection{$\alpha$-Variational Evidence Lower Bound Objective}
Given the large number of parameters and the analytical intractability of our problem, we turn to modern approximation methods. Specifically, we employ Variational Inference (VI) methods \citep{bishop:2006:PRML} in conjunction with empirical Bayes. VI frames marginalization required during Bayesian inference as an optimization problem. This is achieved by assuming the form of the posterior distribution, known as the variational distribution, and performing optimization to find the assumed density closest to the true posterior. This assumption simplifies computation and provides some level of tractability. 

We define the variational distribution for $\bm{\beta} = (\bm{\beta}_1, \cdots, \bm{\beta}_{2p})$ as the following product of independent Gaussian distributions:

\begin{equation}
    \label{eq:variational}
    q(\bm{\beta})=\Pi_{j=1}^{2p}q_j(\bm{\beta}_j) = \Pi_{j=1}^{2p}N(\bm{\beta}_j|\bm{\mu}_j, \Phi_j)
\end{equation}
where $ \left[\bm{\mu}_j\right]_{1 \times d_j} $ is the mean and $   \left[ \Phi_j \right]_{d_j\times d_j} $ is the variance of the variational distribution for each $\bm{\beta}_j$.

Notice that each $\beta_{jd_j}$ shares a common prior variance $r_j^2$ in \eqref{eq:apriori}, while the variational variance for the vector $\bm{\beta}_j$ is a fully learnable $d_j \times d_j$ matrix in \eqref{eq:variational}. The shared prior variance is used to align the objective function with the smoothing spline's Bayesian interpretation, where each $\beta_{jd_j}$ shares a common prior. The full, learnable variational variance simplifies computations during coordinate descent in Section \ref{sec:Coordinate Descent Algorithm}.  Imposing a shared variational variance would complicate the computation process. Given data model \eqref{eq:additive_model_regression_matrix}, prior setup \eqref{eq:prior} and variational setup \eqref{eq:variational}, we have the following objective function: 
\begin{equation} 
    \mathcal{L}(\bm{\mu}, \bm{\Phi}, \bm{r}^2) = \underbrace{- \E_{q} \left[\log p\left(\mathbf{y} \mid \bm{\beta}_1, \cdots, \bm{\beta}_{2p} \right)\right]}_{(I)} + \underbrace{\tilde{\alpha} \cdot \KLD{q}{p}}_{(II)} \label{eq:elbo} 
\end{equation}
%\scalebox{0.9}{ % Adjust the scaling factor as needed
%\begin{minipage}{\textwidth}
%\begin{align}
%    \mathcal{L}(\bm{\mu}, \bm{\Phi}, \bm{r}^2) &= - \E_{q} \left[\log p\left(\mathbf{y} \mid \bm{\beta}_1, \cdots, \bm{\beta}_{2p} \right)\right] + \tilde{\alpha} \cdot \KLD{q}{p} \label{eq:elbo} \\
%    &= - \E_{\underset{j = 1, \dots, 2p}{\bm{\beta}_j \sim \N\left(\bm{\mu}_j, \Phi_j \right)}} \left[ \textstyle \log \N \left(\mathbf{y} \middle| \sum_{j=1}^{2p} Z_j \bm{\beta}_j, \sigma^2 I_n \right) \right] + \tilde{\alpha} \cdot \sum_{j=1}^{2p} \textstyle \KLD{\N \left(\bm{\mu}_j, \Phi_j \right)}{\N \left(\mathbf{0}, r_j^2 I_{d_j} \right)} \label{eq:elbo_spec}
%    \end{align}
%\end{minipage}
%}
where $\bm{\mu} = (\bm{\mu}_1, \cdots, \bm{\mu}_{2p})$, $\bm{\Phi} = (\Phi_1, \cdots, \Phi_{2p})$, and $\bm{r}^2 = (r_1^2, \cdots, r_{2p}^2)$. Our goal is to minimize $\mathcal{L}$ w.r.t. $(\bm{\mu}, \bm{\Phi}, \bm{r}^2)$. One can then use variational means $\bm{\mu}_1, \cdots, \bm{\mu}_{2p}$ as estimators for coefficients to make predictions.

This objective function is closely related to the expected lower bound (ELBO). The first term (I) represents the expected log-likelihood of the data with $\E_{q}$ denoting the expectation with respect to the variational distribution $q(\bm{\beta})$ defined in \eqref{eq:variational}. The second
term (II) represents the  KL divergence between the variational posterior $q$ and the prior $p$. Ignoring constants, we can write down the detailed expressions of each term in equation \eqref{eq:elbo} as follows:
%\scalebox{0.86}{ % Adjust the scaling factor as needed
%\begin{minipage}{\textwidth}
\begin{align}
    (I) & = - \E_{\underset{j = 1, \dots, 2p}{\bm{\beta}_j \sim \N\left(\bm{\mu}_j, \Phi_j\right)}} \left[\log \N\left(\mathbf{y} \middle| \sum_{j=1}^{2p} Z_j \bm{\beta}_j, \sigma^2 I_n\right)\right] \nonumber \\
     &\overset{c}{=} \frac{1}{2\sigma^2} \left\{\left\|\mathbf{y} - \sum_{j=1}^{2p} Z_j \bm{\mu}_j \right\|_2^2 + \sum_{j=1}^{2p} \operatorname{tr}\left(\Phi_j Z_j^T Z_j\right)\right\} %+ \frac{n}{2} \log \sigma^2 
     \label{eq:elbo_spec_loss_term} \\
     (II) & = 
    \tilde{\alpha} \cdot \sum_{j=1}^{2p} \KLD{\N\left(\bm{\mu}_j, \Phi_j\right)}{\N\left(\mathbf{0}, r_j^2 I_{d_j}\right)} \nonumber \\
    &\overset{c}{=} \frac{\tilde{\alpha}}{2} \sum_{j=1}^{2p} \left\{d_j \log r_j^2 - \log \det \Phi_j + \frac{\left\|\bm{\mu}_j\right\|_2^2 + \operatorname{tr}\left(\Phi_j\right)}{r_j^2}\right\} \label{eq:elbo_spec_penalty_term}
\end{align}
%\end{minipage}
%}

% \fl{do we need this ? }
% Here $\overset{c}{=}$ means equal to up to multiplication or add constant, $tr(\cdot)$ is the trace of a matrix and $\det$ is the determinant of a matrix. 
 We discuss the three key elements of our objective function below:

\begin{itemize}
    \item \textbf{Emprical Bayes Prior:} In a typical ELBO, the prior is pre-specified. In our framework, however, we estimate the prior variances $r_j^2$ alongside the variational parameters by optimizing \eqref{eq:elbo}. As mentioned before, the feature selection is performed naturally whenever some of the $r_j^2$ reach exactly zero. We will later observe (see Section \ref{sec:setup}) that $r_j^2 = 0$ implies a point mass variational distribution at zero ($\bm{\mu}_j = \mathbf{0}, \Phi_j = \mathbf{0}$) as well. Thus, feature selection can be equivalently performed by checking the sparsity of the estimated $\bm{\mu}_1, \cdots, \bm{\mu}_{2p}$.

    \item \textbf{$\alpha$-variational inference:} In \eqref{eq:elbo}, we've introduced a positive hyperparameter $\tilde{\alpha} > 0$. This concept of adjusting the weight of the KL divergence, known as $\alpha$-variational inference \citep{yang2020alpha}, has been employed previously in the context of variational autoencoder \citep{higgins2017betavae}. Similar to the tuning parameter in Lasso, $\tilde{\alpha}$ can be adjusted through cross-validation for optimal model performance. Larger values of $\tilde{\alpha}$ tend to encourage a sparser model and vice versa.  When $\tilde{\alpha} = 1$, the objective function \eqref{eq:elbo} is equivalent to the original ELBO function in variational inference.

    \item \textbf{Working with Unknown $\sigma^2$:}  Recall that our objective function \eqref{eq:elbo} is the sum of equations (\ref{eq:elbo_spec_loss_term}) and \eqref{eq:elbo_spec_penalty_term}. Note that, when treated as fixed hyperparameters, $\tilde{\alpha}$ and $\sigma^2$ are inherently coupled. Since we already need to tune $\tilde{\alpha}$, we can circumvent the direct estimation of $\sigma^2$ by consolidating these two hyperparameters into a new hyperparameter $\alpha:= \tilde{\alpha}\sigma^2$, thereby simplifying equation  \eqref{eq:elbo} as follows:
\end{itemize}

\scalebox{0.9}{ % Adjust the scaling factor as needed
\begin{minipage}{\textwidth}
\begin{align}
&     \mathcal{L}  = (I) + (II) \nonumber \\
 %   & \propto - \E_{\underset{j = 1, \dots, 2p}{\bm{\beta}_j \sim \N\left(\bm{\mu}_j, \Phi_j\right)}} \left[\log \N\left(\mathbf{y} \middle| \sum_{j=1}^{2p} Z_j \bm{\beta}_j, I_n\right)\right] + \tilde{\alpha}\sigma^2 \sum_{j=1}^{2p} \KLD{\N\left(\bm{\mu}_j, \Phi_j\right)}{\N\left(\mathbf{0}, r_j^2 I_{d_j}\right)} \nonumber \\
 %   & \propto - \E_{\underset{j = 1, \dots, 2p}{\bm{\beta}_j \sim \N\left(\bm{\mu}_j, \Phi_j\right)}} \left[\log \N\left(\mathbf{y} \middle| \sum_{j=1}^{2p} Z_j \bm{\beta}_j, \tilde{\alpha}\sigma^2 I_n\right)\right] + \sum_{j=1}^{2p} \KLD{\N\left(\bm{\mu}_j, \Phi_j\right)}{\N\left(\mathbf{0}, r_j^2 I_{d_j}\right)} \nonumber \\
    &\overset{c}{=} \frac{1}{\tilde{\alpha}\sigma^2} \left\{\left\|\mathbf{y} - \sum_{j=1}^{2p} Z_j \bm{\mu}_j \right\|_2^2 + \sum_{j=1}^{2p} \operatorname{tr}\left(\Phi_j Z_j^T Z_j\right)\right\} + \sum_{j=1}^{2p} \left\{d_j \log r_j^2 - \log \det \Phi_j + \frac{\left\|\bm{\mu}_j\right\|_2^2 + \operatorname{tr}\left(\Phi_j\right)}{r_j^2}\right\} \label{eq:elbo_alter_1} \\
    &\overset{c}{=} \boxed{\frac{1}{\alpha} \left\{\left\|\mathbf{y} - \sum_{j=1}^{2p} Z_j \bm{\mu}_j \right\|_2^2 + \sum_{j=1}^{2p} \operatorname{tr}\left(\Phi_j Z_j^T Z_j\right)\right\} + \sum_{j=1}^{2p} \left\{d_j \log r_j^2 - \log \det \Phi_j + \frac{\left\|\bm{\mu}_j\right\|_2^2 + \operatorname{tr}\left(\Phi_j\right)}{r_j^2}\right\}} \label{eq:elbo_alter_2}
\end{align}
\end{minipage}
}

\eqref{eq:elbo_alter_2} will be our final objective function.

\section{Coordinate Descent Algorithm}
\label{sec:Coordinate Descent Algorithm}
In this section, we introduce our blockwise coordinate descent algorithm for optimizing \eqref{eq:elbo_alter_2} with respect to prior variances (smoothness parameter) $\{r_j\}_{j=1}^{2p}$, variational means $\{\bm{\mu}_j\}_{j=1}^{2p}$ and variational variances $\{\Phi_j\}_{j=1}^{2p}$ given a fixed positive hyperparameter $\alpha$. We will elaborate on how our method can reach exact sparsity and also discuss how to tune $\alpha$ efficiently.

\subsection{Algorithm}
The core idea of our coordinate descent algorithm can be summarized as follows. We start by initializing a set of prior variances, variational means, and variational variances $\{r_j^2, \bm{\mu}_j, \Phi_j\}_{j=1}^{2p}$; then sequentially optimize the objective function \eqref{eq:elbo_alter_2} by focusing on one parameter block at a time, which includes $\big(r_j^2, \bm{\mu}_j, \Phi_j\big)$, while keeping all other parameters fixed. We iterate this sequential updating procedure until a convergent solution is achieved. %%Below is a comprehensive description of the optimization process for a single block of $(r_j^2, \bm{\mu}_j, \Phi_j)$:

To optimize one block of $(r_j^2, \bm{\mu_}j, \Phi_j)$ while keeping all other parameters fixed, we can see that \eqref{eq:elbo_alter_2} simplifies to the following marginal objective function w.r.t $(r_j^2, \bm{\mu_}j, \Phi_j)$:
\begin{align}
   & \mathcal{L}_j \left(r_j^2, \bm{\mu}_j, \Phi_j\right)    \nonumber \\
   \overset{c}{=}  & \frac{1}{\alpha} \left\{\left\|\mathbf{y}_{(-j)} - Z_j \bm{\mu}_j \right\|_2^2 + \operatorname{tr}\left(\Phi_j Z_j^T Z_j\right)\right\} + \left\{d_j \log r_j^2 - \log \det \Phi_j + \frac{\left\|\bm{\mu}_j\right\|_2^2 + \operatorname{tr}\left(\Phi_j\right)}{r_j^2}  \right\} \label{eq:one_block_problem}
\end{align}
%\scalebox{0.85}{
%\begin{minipage}{\linewidth}
%\begin{align}
%    \mathcal{L}_j \left(r_j^2, \bm{\mu}_j, \Phi_j\right)    &\overset{c}{=} \frac{1}{\alpha} \left\{\left\|\mathbf{y}_{(-j)} - Z_j \bm{\mu}_j \right\|_2^2 + \operatorname{tr}\left(\Phi_j Z_j^T Z_j\right)\right\} + \left\{d_j \log r_j^2 - \log \det \Phi_j + \frac{\left\|\bm{\mu}_j\right\|_2^2 + \operatorname{tr}\left(\Phi_j\right)}{r_j^2}  \right\} \label{eq:one_block_problem}
%\end{align}
%\end{minipage}
%}
where, $\mathbf{y}_{(-j)} := \mathbf{y} - \sum_{{j'}\neq j} Z_{j'}\bm{\mu}_{j'}$.

%%Another way to interpret objective \eqref{eq:one_block_problem} is to view it as the original problem \eqref{eq:elbo_alter_2} with only one block.

%%\subsubsection{Parameter Reduction}

While optimizing $\mathcal{L}_j$ might initially appear challenging, it turns out that, for a fixed $r_j^2$, one can easily determine the optimal values of $\bm{\mu}_j$ and $\Phi_j$ as functions of $r_j^2$. By substituting the optimal expressions for $\bm{\mu}_j$ and $\Phi_j$ as functions of $r_j^2$ into \eqref{eq:one_block_problem}, we can simplify the problem to optimizing the following univariate objective function \(G_{\alpha, \bm{\eta}_j, V_j}(r_j^2)\). This function depends on the tuning parameter \(\alpha\), the \(j^{\text{th}}\) Gram matrix \(V_j := Z_j^T Z_j\), and \(\bm{\eta}_j := Z_j^T \mathbf{y}_{(-j)}\).
\begin{equation}
\label{eq:univariate_objective}
G_{\alpha, \bm{\eta}_j, V_j}(r_j^2) := \sum_{k=1}^{d_j} \left\{ \alpha \log \left( v_{jk} r_j^2 + \alpha \right) - \frac{\eta_{jk}^2 r_j^2}{v_{jk} r_j^2 + \alpha} \right\}
\end{equation}
Here, \(v_{jk}\) is the \(k^{\text{th}}\) diagonal element of the matrix \(V_j\), and \(\eta_{jk}\) is the \(k^{\text{th}}\) element of the vector \(\bm{\eta}_j\). We summarize this result with the following proposition, and the proof can be found in Appendix \ref{sec:appendix_parameter_reduction}.

%%In particular, when $r_j^2$ remains fixed, the solution for the optimal variational distribution $N(\hat{\bm{\mu}}_j, \hat{\Phi}_j)$ coincides with the posterior distribution of $\bm{\beta}_j$ obtained when using the model $\mathbf{y}_{(-j)}|\bm{\beta}_j \sim N(Z_j\bm{\beta}_j, \alpha I_n)$ with a prior distribution of $\bm{\beta}_j \sim N(\mathbf{0}, r_j^2I_{d_j})$. We summarize this result with the following proposition, and the proof can be found in Appendix \ref{sec:appendix_parameter_reduction}. Proposition \ref{prop:optimal_mu_phi} is of utmost importance, as it implies that the three-dimensional optimization problem of \eqref{eq:one_block_problem} can be simplified to an optimization problem solely dependent on the univariate variable $r_j^2$.

\begin{prop}
\label{prop:reduce_to_univariate_problem}
$(\hat{r}_j^2, \hat{\bm{\mu}}_j, \hat{\Phi}_j)$ minimize objective \eqref{eq:one_block_problem} if only if
\begin{align}
    \hat{r}_j^2 &= \underset{r^2 \geq 0}{\argmin} \quad G_{\alpha, \bm{\eta}_j, V_j}(r^2)\label{eq:univariate_problem}\\
    \hat{\bm{\mu}}_j &= \hat{r}_j^2 \cdot (\hat{r}_j^2Z_j^TZ_j + \alpha I_{d_j})^{-1}Z_j^T\mathbf{y}_{(-j)} \label{eq:mu_expression}\\
    \hat{\Phi}_j &= \hat{r}_j^2 \cdot \alpha (\hat{r}_j^2Z^T_jZ_j + \alpha I_{d_j})^{-1}\label{eq:phi_expression}
\end{align}
\end{prop}

%%\paragraph{Univariate Optimization Problem} Substituting the expressions for the optimal $\hat{\bm{\mu}}_j$ and $\hat{\Phi}_j$ from proposition \ref{prop:optimal_mu_phi} back into the objective \eqref{eq:one_block_problem}, we find that $\mathcal{L}_j$ can be further reduced to the following problem, where we seek the optimal value of $r_j^2$ for a univariate function depending on the tuning parameter $\alpha$, the $j^{th}$ gram matrix $V_j:= Z_j^TZ_j$, and $\bm{\eta}_j:= Z_j^T\mathbf{y}_{(-j)}$. The derivation details can be found in Appendix \ref{sec:appendix_parameter_reduction}:

%%\begin{equation}
%%\label{eq:univariate_problem}
%%    \hat{r}_j^2 = \underset{r^2 \geq 0}{\argmin} \quad G_{\alpha, \bm{\eta}_j, V_j}(r^2)
%%\end{equation}

%%The univariate function $G_{\alpha, \bm{\eta}_j, V_j}(r^2)$ has the following additive form:

%%\begin{equation}
%%\label{eq:additive_form_univariate_problem}
%%    G_{\alpha, \bm{\eta}_j, V_j}(r^2):=\sum_{k=1}^{d_j} g_{\alpha, \eta_{jk}, v_{jk}}(r^2)
%%\end{equation}

%%Where each term in the sum is defined as:

%%\begin{equation}
%%    g_{\alpha, \eta_{jk}, v_{jk}}(r^2) := \alpha \log (v_{jk} r^2 + \alpha) - \frac{\eta_{jk}^2r^2}{v_{jk}r^2 + \alpha}
%%\end{equation}

%%Here, $v_{jk}$ is the $k^{th}$ diagonal element of the matrix $V_j$, and $\eta_{jk}$ is the $k^{th}$ element of the vector $\bm{\eta}_j$.
%%\subsubsection{Optimize the univariate problem}
%%\label{sec:optimize_univariate_problem}
Proposition \ref{prop:reduce_to_univariate_problem} is of utmost importance, as it implies that the three-dimensional optimization problem of \eqref{eq:one_block_problem} can be simplified to a univariate optimization problem \eqref{eq:univariate_problem}. Unfortunately, the global minimum of the data-dependent function $G_{\alpha, \bm{\eta}_j, V_j}(r^2)$ is not generally available in closed form unless $d_j = 1$. The function itself can be non-convex. However, by analyzing the form of $G_{\alpha, \bm{\eta}_j, V_j}(r^2)$ in \eqref{eq:univariate_objective}, we discover that it is possible to narrow down potential range of the global minimum from the entire interval $r^2 \in [0, \infty)$ to a much smaller closed interval $[l_j, u_j]$, where
\begin{align}
    l_j := \underset{k = 1, \cdots, d_j}{\min}
    \left\{\left(\frac{\eta^2_{jk} - \alpha v_{jk}}{v_{jk}^2}\right)_{+}\right\}; 
    u_j := \underset{k = 1, \cdots, d_j}{\max}\left\{\left(\frac{\eta^2_{jk} - \alpha v_{jk}}{v_{jk}^2}\right)_{+}\right\}.
\end{align}
We formally state this property as the following proposition, with the proof provided in Appendix \ref{sec:appendix_univariate_solution}.

\begin{prop}
\label{prop:univariate_problem_interval}
The optimal solution for univariate problem \eqref{eq:univariate_problem} exists and must be within closed interval $[l_j, u_j]$.
\end{prop}
Computationally, it is now possible to simply grid search $G_{\alpha, \bm{\eta}_j, V_j}(r^2)$ on $[l_j, u_j]$ to find $\hat{r}_j^2$. In practice, we found a grid number $1000\times d_j$ works well enough.

Apart from computational convenience, Proposition \ref{prop:univariate_problem_interval} also gives us insight into how sparsity is reached in our algorithm. As a reminder, by exact sparsity in our framework we mean $\hat{r}_j^2 = 0$. This also means $(\hat{\bm{\mu}}_j, \hat{\Phi}_j) = (\mathbf{0}, \mathbf{0})$ as a direct consequence of proposition \ref{prop:reduce_to_univariate_problem}. Further, in Proposition \ref{prop:univariate_problem_interval}, the endpoints of the interval $[l_j, u_j]$ correspond to the minimum and maximum values of the sequence $\left\{\left(\frac{\eta^2_{jk} - \alpha v_{jk}}{v_{jk}^2}\right)_{+}\right\}_{k=1}^{d_j}$. This adds interpretability between the relationship of $\alpha$ and sparsity. This is summarized in Table \ref{tab:optimal_rj} and illustrated in Figure \ref{fig:G_function}.

\begin{table}[h]
\centering
\scalebox{0.85}{
\begin{tabular}{p{7.5cm}|p{4cm}|p{5cm}}
\toprule
Condition & $(l_j, u_j)$ & Optimal $r_j^2$ \\
\midrule
$\alpha < \underset{k = 1, \cdots, d_j} {\min}\left(\frac{\eta_{jk}^2}{v_{jk}}\right)$ (left column in Fig \ref{fig:G_function})& $l_j > 0$ & Nonzero \\
\midrule
$\alpha \geq \underset{k = 1, \cdots, d_j}{\max}\left(\frac{\eta_{jk}^2}{v_{jk}}\right)$ (right column in Fig \ref{fig:G_function})& $l_j = u_j = 0$ & Exactly zero \\
\midrule
$\underset{k = 1, \cdots, d_j}{\min}\left(\frac{\eta_{jk}^2}{v_{jk}}\right) \leq \alpha < \underset{k = 1, \cdots, d_j}{\max}\left(\frac{\eta_{jk}^2}{v_{jk}}\right)$ (middle column in Fig \ref{fig:G_function}) & $l_j = 0$ & Depends on $\alpha$, $\bm{\eta}_j$, and $V_j$ \\
\bottomrule
\end{tabular}
}
\caption{Conditions for determining the optimal value of $r_j^2$.}
\label{tab:optimal_rj}
\end{table}

\begin{figure}[h]
\centering
{\includegraphics[width=0.8\textwidth]{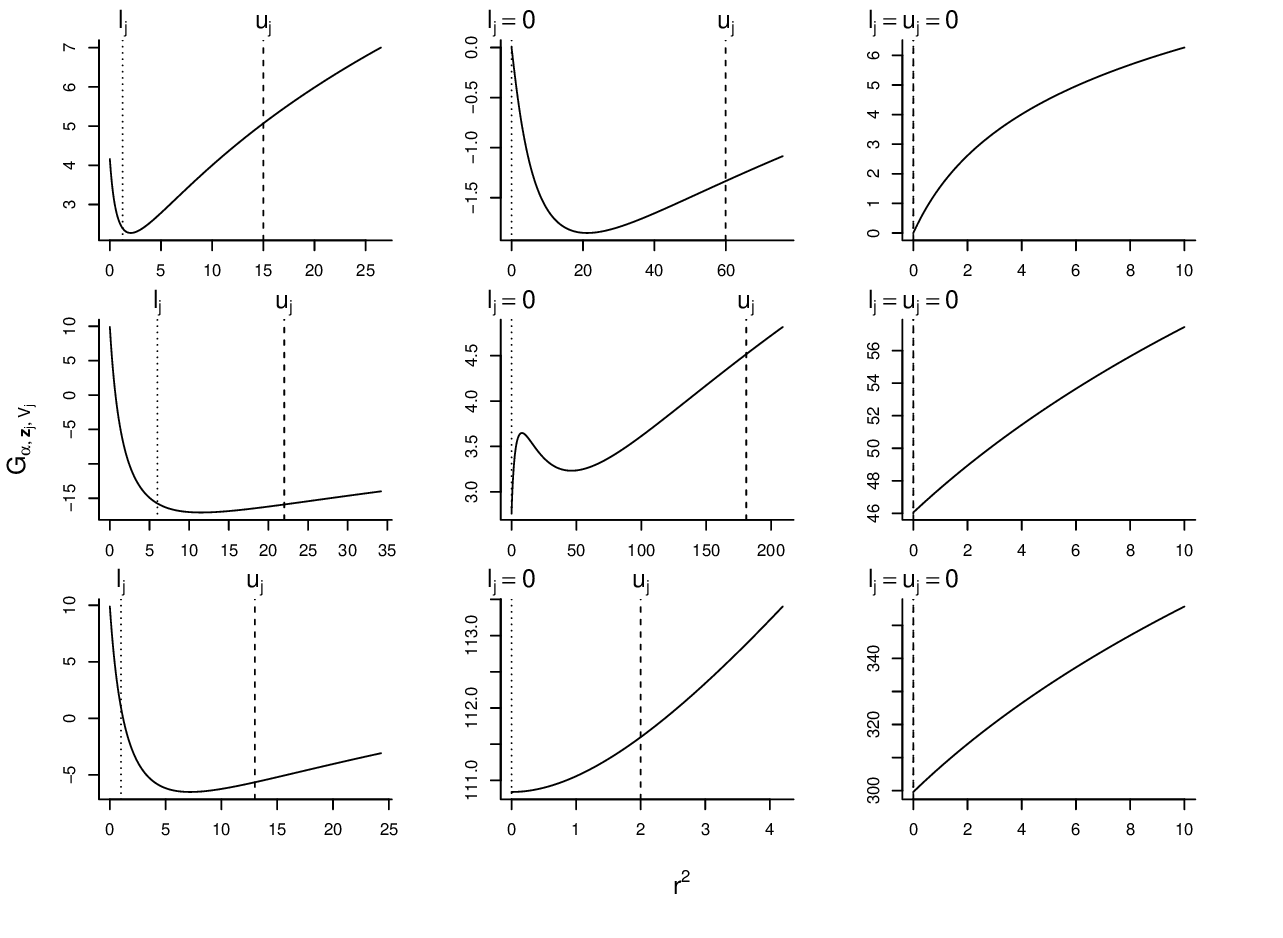}}
\caption[short]{Example of possible shapes of function $G_{\alpha, \bm{\eta}_j, V_j}(r^2)$ under different $\alpha$, $\bm{\eta}_j$, $V_j$ value. Left column: $\alpha < \underset{k = 1, \cdots, d_j}{min}(\frac{\eta_{jk}^2}{v_{jk}})$,  middle column: $ \underset{k = 1, \cdots, d_j}{min}(\frac{\eta_{jk}^2}{v_{jk}}) \leq \alpha < \underset{k = 1, \cdots, d_j}{max}(\frac{\eta_{jk}^2}{v_{jk}})$, right column: $\alpha \geq \underset{k = 1, \cdots, d_j}{max}(\frac{\eta_{jk}^2}{v_{jk}})$}.
\label{fig:G_function}
\end{figure}

With the analysis above, we present our coordinate descent algorithm as follows:

\begin{algorithm}[H]
    \label{alg:coordinate_descent}
    \caption{Coordinate Descent}
    Input $\mathbf{y}, \{Z_j\}_1^{2p}$ defined in Section \ref{sec:setup} \;
    Choose hyperparameter $\alpha > 0$ \;
    Init $\hat{\bm{\mu}}_1, \cdots, \hat{\bm{\mu}}_{2p}$ \;
    \While{Not Converge}{
    \For{$j$ in $1,\dots, 2p$}{
    $\mathbf{y}_{(-j)} = \mathbf{y} - \sum_{{j'}\neq j} Z_{j'}\bm{\mu}_{j'}$ \;
    $\bm{\eta}_j = Z_j^T\mathbf{y}_{(-j)}$ \;
    $\hat{r}_j^2  = \underset{r^2 \in [l_j, u_j]}{\argmin} \, G_{\alpha, \bm{\eta}_j, V_j}(r^2)$ \;
    $\hat{\bm{\mu}}_j = \hat{r}_j^2 \cdot (\hat{r}_j^2Z_j^TZ_j + \alpha I_{d_j})^{-1}Z_j^T\mathbf{y}_{(-j)}$  \;
    $\hat{\Phi}_j = \hat{r}_j^2 \cdot \alpha (\hat{r}_j^2Z^T_jZ_j + \alpha I_{d_j})^{-1}$
  }
  }
\end{algorithm}

\subsection{Cross Validation to Choose $\alpha$}
The hyperparameter $\alpha$ is essential for controlling sparsity in the estimated coefficients, with larger values promoting greater sparsity. In fact, from the previous discussion, when $\alpha$ exceeds a threshold $\underset{j = 1, \cdots, 2p}{\max}\{{\underset{k = 1, \cdots, d_j}{\max}\frac{(Z_j^T\mathbf{y})_k^2}{v_{jk}}}\}$, all coefficients initialized at zero will remain zero, similar to the penalty parameter $\lambda$ in \texttt{glmnet}. To determine the optimal $\alpha$, we perform cross-validation over a predefined range and utilize a warm start approach for efficiency. A U-shaped cross-validation error curve is typically observed. To mitigate overfitting, we select $\alpha$ as the maximum value within 0.15 times the standard deviation of the minimum cross-validation error, we call this $\alpha_{0.15se}$.  Further details about the cross validation process and other implementation details are provided in the Appendix \ref{sec:appendix_alpha_implementation}.

\section{Experiments}
In this section, we present a series of experiments involving both synthetic and real-world datasets, showcasing the effectiveness of our approach in providing feature selection and individual smoothing as compared to other state-of-the-art methods. 

We begin by applying our algorithm to the Boston Housing Dataset in Section \ref{sec:boston}, illustrating the step-by-step procedure, including hyperparameter optimization through cross-validation to show which features offer more relevant information in predicting the response. Following this, in Section \ref{sec:smoothness}, we demonstrate how our method adaptively determines the smoothness of each function $f_j$ and compare its performance against other methods including SPAM and GAMSEL using simulated data. Finally, in Section \ref{sec:performance_comparison}, we evaluate the prediction and selection accuracy of our approach relative to SPAM and GAMSEL. Additionally, we assess our method's capability to accurately categorize $f_j$ as either a zero, linear, or nonlinear function in contrast to GAMSEL. 

The experimental results from SPAM and GAMSEL are evaluated using their R package implementation \texttt{SAM} \citep{SAM_R_Package} and \texttt{gamsel} \citep{GAMSEL_R_Package}, respectively. While we attempted to include COSSO to compare in our experiment, we encountered technical difficulties by using their recent R package implementation \texttt{cosso} \citep{COSSO_R_Package}.

\subsection{Boston Housing Dataset}
\label{sec:boston}
\begin{figure}[!ht]
\centering
{\includegraphics[width=\textwidth]{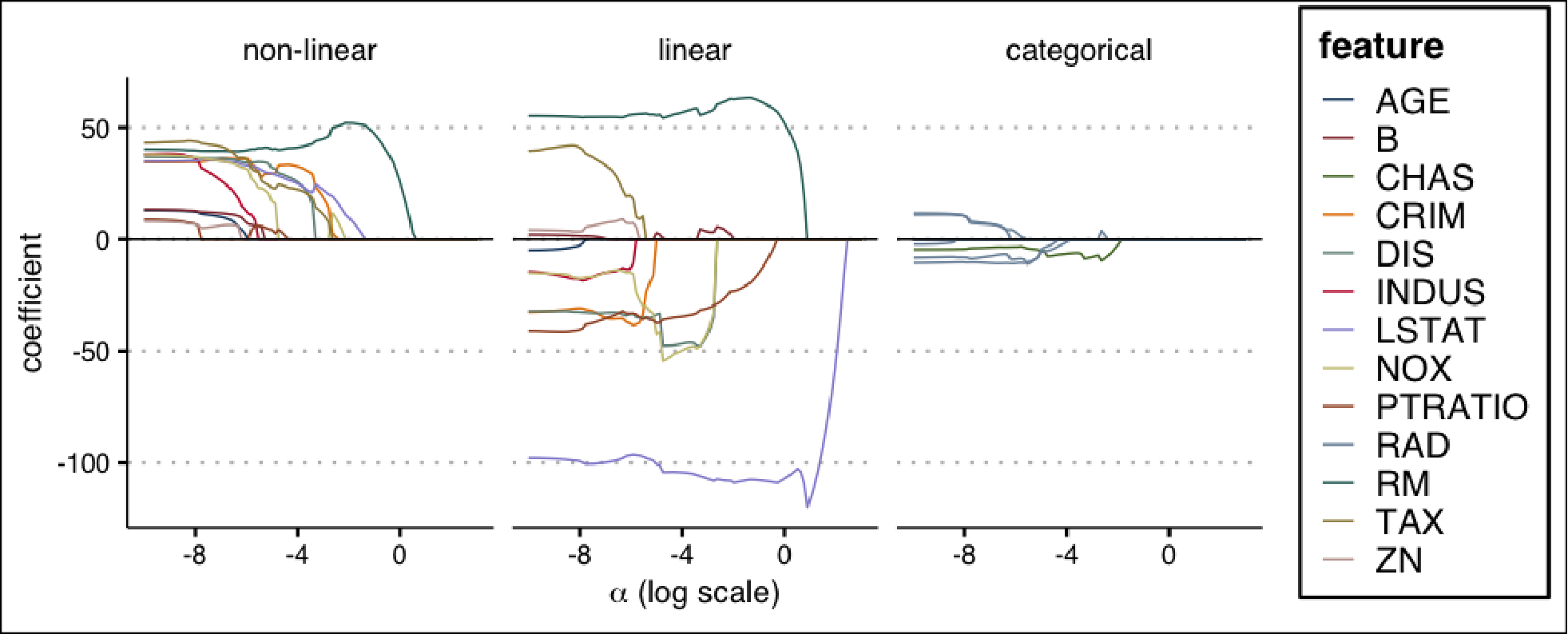}}
\caption[short]{The Model Path for the Boston Housing Dataset. Here, the coefficients $\hat{\bm{\mu}}_j$ for each term $Z_j$ are adjusted as if $Z_j^TZ_j = I_{d_j}$. For non-linear terms, we plot $\|\hat{\bm{\mu}}_j\|_2$ instead. Each numerical feature has a non-linear and linear term, both depicted in the same color. In the categorical panel, each line corresponds to a level of the categorical feature, with levels from the same feature sharing the same color.}
\label{fig:boston_model_path}
\end{figure}
We start by applying our framework to the Boston Housing Dataset, a dataset previously utilized in the works of \cite{alex2015} and \cite{ravikumar2007spam}. This dataset comprises 506 housing records in the Boston area, with the target variable being MEDV (Median value of owner-occupied homes in $1000$'s), and it includes 13 covariates. Among these covariates, we treat CHAS and RAD as categorical variables due to them having fewer than 10 unique values. Additionally, we perform a log transformation on CRIM and DIS, as these two variables exhibit significant right skewness. For the remaining 11 numerical features, we employ a natural cubic spline with 10 knots based on quantiles. Together with the two one-hot encoded categorical features, we execute our model across a sequence of 100 $\alpha$ values evenly distributed in log scale from $-10$ to $3$. The model path is depicted in Figure \ref{fig:boston_model_path}, where an increase in the value of $\alpha$ eventually leads to all term coefficients reaching precisely zero.

\begin{figure}[!ht]
\centering
{\includegraphics[width=0.8\textwidth]{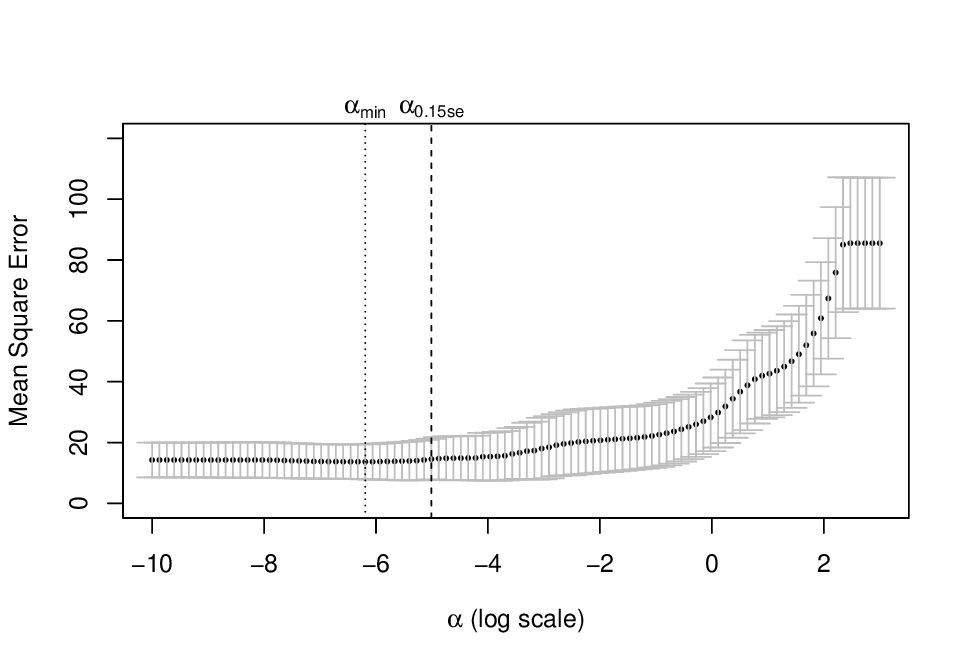}}
\caption[short]{The 10-fold Cross Validation results for the Boston Housing Dataset. Each black point represents the average cross-validation mean squared error for a specific $\alpha$ value, while the grey bars indicate one standard deviation away from the average.}
\label{fig:boston_cv}
\end{figure}

To select the optimal $\alpha$, we conducted a 10-fold cross-validation, resulting in the recommended value $\alpha_{0.15se}$. The cross-validation results are illustrated in Figure \ref{fig:boston_cv}, where the average cross-validation mean squared error for $\alpha_{min}$ and $\alpha_{0.15se}$ are 13.68 and 14.54, respectively.

With $\alpha_{0.15se}$, our model identifies ZN, AGE, and B as irrelevant features (achieving exact sparsity on both linear and non-linear terms), PTRATIO as a linear feature (achieving exact sparsity on the non-linear term but having non-zero linear coefficients), and the rest as non-linear features (having non-zero non-linear coefficients). Among the two categorical features, CHAS only has one non-zero level -- `0', while RAD has non-zero levels `1', `2', and `7'. For detailed fitted curves of numerical features and coefficients of categorical features, please refer to Figure \ref{fig:boston_fitted}.

\begin{figure}[!ht]
\centering
{\includegraphics[width=\textwidth]{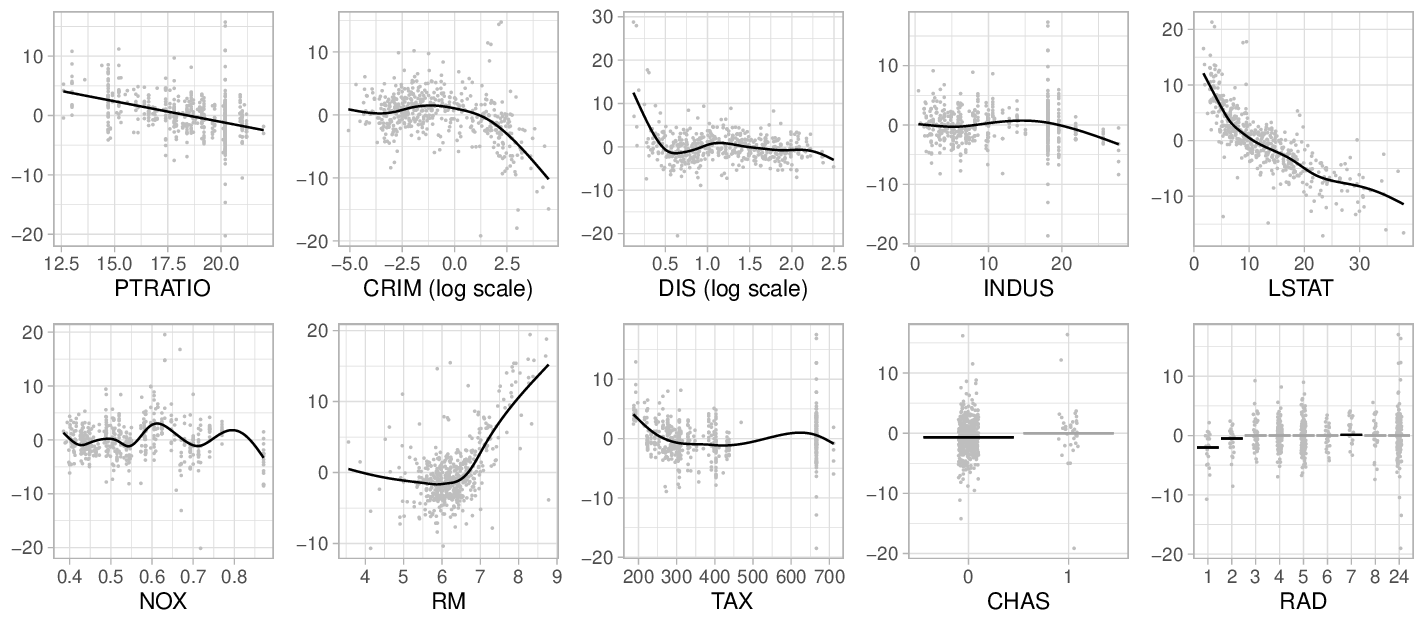}}
\caption[short]{Estimated non-zero $f_j$ for each feature in the Boston Housing Dataset using $\alpha_{0.15se}$. A grey dot on the y-axis represents the response value after removing the fitted value for the rest of the features. For categorical features CHAS and RAD, the y-axis is shifted by the fitted value of the level with a zero coefficient, a grey line indicates that the fitted coefficient for that level is exactly zero. For the feature PTRATIO, the fitted $f_j$ is exactly linear. ZN, AGE and B are not presented in this plot as our model identifies those as irrelevant features.}
\label{fig:boston_fitted}
\end{figure}

\subsection{Fitting the Smoothness}
\label{sec:smoothness}
In this section, we illustrate that our method can find the proper smoothness for every feature using simulated data and compare it against other methods including SPAM and GAMSEL.
\subsubsection{One Feature Case}
To start, let us consider the following simple dataset with only one feature,
$$y_i = f(x_i) + \epsilon_i, \quad i = 1, \cdots, n$$
Here, the sample size is $n = 300$, $\epsilon_i$ are sampled from a standard Gaussian distribution, and $x_i$ are uniformly sampled from the interval $(0,1)$. The ground truth function $f$ is defined as:
\begin{equation}
\label{eq:smoothness_1d_ground_truth}
    f(x) = e^x\cdot\sin\{13(x-0.23)^2\} + 5x, \quad x \in (0,1).
\end{equation}
We conducted experiments using SPAM, GAMSEL, and our VARD model on a simulated dataset with varying penalty hyperparameters ($\lambda$ for SPAM and GAMSEL, $\alpha$ for VARD). To ensure a fair comparison, we kept the number of basis spline functions consistent across all models, with parameters set as follows: p = 25 in SPAM, dim = 25 in GAMSEL, and 25 knots in our VARD model. For our model VARD, we use natural cubic spline basis. For illustrating the effect of $\lambda$ on smoothness in GAMSEL, we set the degree of freedom to the maximum of 25. We used the default $\gamma = 0.4$ for GAMSEL. The results are visualized in Figure \ref{fig:smoothness_1d}, where the penalty hyperparameter gradually increases from left to right.

In Figure \ref{fig:smoothness_1d}, subplots with dashed borders represent the fitted curve of $\lambda_{min}$ or $\alpha_{min}$ for different methods in ten-fold cross-validation, while subplots with solid borders show the fitted curve of $\lambda_{1se}$ or $\alpha_{0.15se}$ for different methods. For SPAM, we didn't display the best model using the GCV (generalized cross-validation) or Cp metric, as recommended by the authors, dut to two reasons: 1) In one-dimensional cases, it can be proven that the best model according to GCV is always the one with no penalty ($\lambda = 0$), which typically leads to overfitting when the basis number is high; 2) Estimating error variance, required for Cp, was not provided in the SPAM paper.

\begin{figure}[!ht]
\centering
{\includegraphics[width=1.01\textwidth]{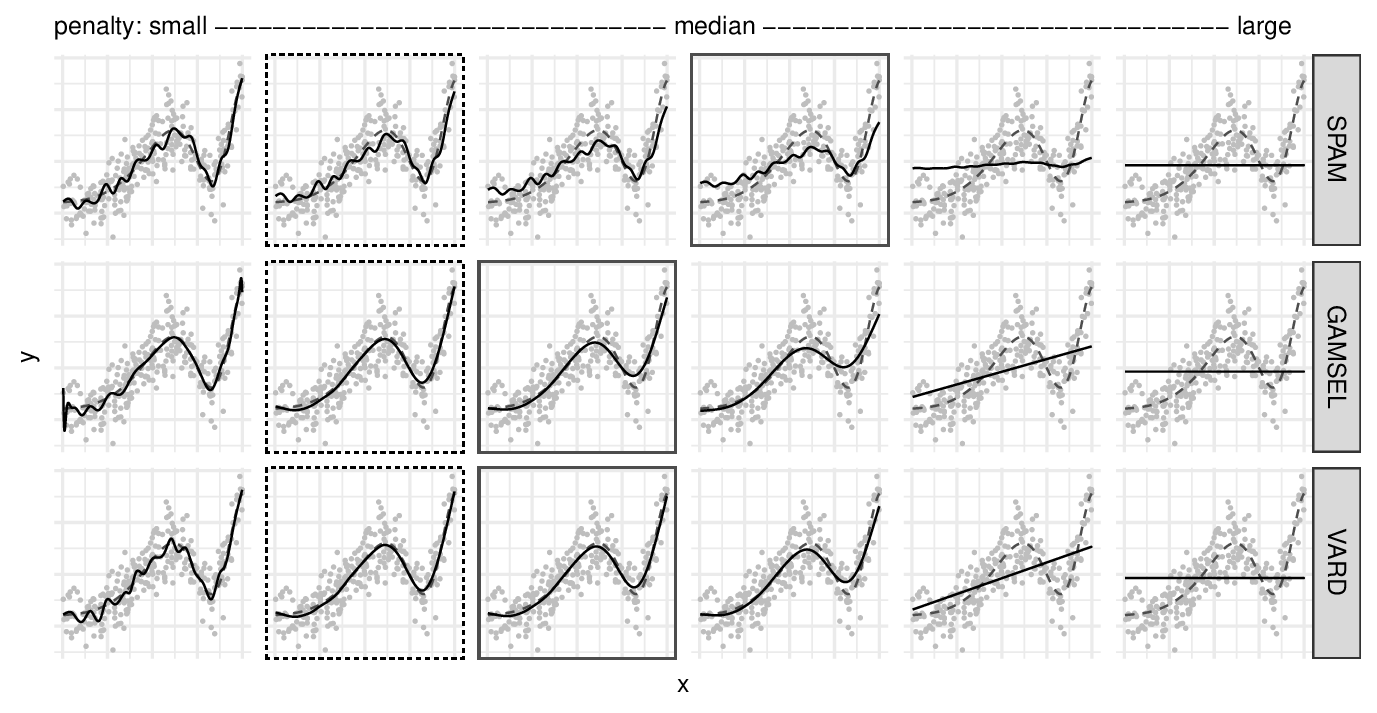}}
\caption[short]{Simulation study on fitting the smoothness for method SPAM, GAMSEL, and VARD in one feature case. The solid curves represent the fitted $\hat{f}(x)$ for each method, with varying penalty parameters ($\lambda$ for SPAM and GAMSEL, $\alpha$ for VARD) increasing from left to right. The dashed curves depict the ground truth $f(x)$ defined in \eqref{eq:smoothness_1d_ground_truth}, while the grey dots indicate the simulated data points ${(x_i, y_i)}_1^n$. In the subplots with dashed borders, we can see the fitted curves for $\lambda_{min}$ (for SPAM and GAMSEL) and $\alpha_{min}$ (for VARD) as determined through ten-fold cross-validation. In the subplots with solid borders, the fitted curves correspond to $\lambda_{1se}$ (for SPAM and GAMSEL) and $\alpha_{0.15se}$ (for VARD). For all three methods, the fitted curves in the rightmost column are exactly zero. In the second-rightmost column, both GAMSEL and VARD produce fitted curves that are exactly linear.}
\label{fig:smoothness_1d}
\end{figure}
From Figure \ref{fig:smoothness_1d} we observe that:
\begin{itemize}
    \item The SPAM method exhibits the capability to achieve exact sparsity but lacks any significant smoothing effect. It's worth noting how the shape of the fitted curve remains mostly unchanged even for relatively large penalty values. This behavior is expected because SPAM penalizes the magnitude of the function rather than its roughness (see Table \ref{tb:regularization_methods}). Mathematically, SPAM is equivalent to shrinking the least square estimator by a multiplier when there's only one feature \citep{ravikumar2009}. Essentially, this implies that if the number of basis spline functions isn't initially set to match the smoothness of the true function, SPAM will be unable to adapt its smoothness to the truth, regardless of how the penalty parameter $\lambda$ is tuned.
    
    \item For both GAMSEL and our method VARD, when the penalty parameter is small, the fitted curves exhibit a high degree of oscillation. As the penalty parameter increases, the fitted curves become smoother, eventually reaching exact linearity and, ultimately, exact sparsity. Both $\lambda_{min}$ and $\lambda_{1se}$ for GAMSEL, as well as $\alpha_{min}$ and $\alpha_{0.15se}$ for VARD, closely match the ground truth. Both GAMSEL and VARD have the ability to capture the correct smoothness in the one-feature case, provided that the models are initialized with sufficient complexity (i.e., a large enough number of basis functions and degrees of freedom for GAMSE, or a large enough number of basis functions for VARD).

    \item The SPAM method does not exhibit an exact linear phase, but both GAMSEL and VARD can distinguish whether a function is exactly zero, linear, or non-linear.
\end{itemize}
\subsubsection{Multiple Feature Case}
\label{sec:Multiple_Feature_Case}
While SPAM evidently falls short in terms of fitting smoothness, GAMSEL performs adequately and closely competes with VARD, in single-feature scenarios. However, as we will demonstrate next, GAMSEL encounters challenges in precisely determining the appropriate smoothness for each feature individually in multi-feature contexts. This issue stems from the influence of the sparsity-inducing hyperparameter $\lambda$ in GAMSEL, which has a global impact on the smoothness of all features (see Table \ref{tb:regularization_methods}). In contrast,  VARD excels in achieving the appropriate smoothness for multiple features. This is because our objective function  \eqref{eq:elbo_alter_2} learns the smoothness of each feature independently and data-adaptively using $r_1^2, \cdots, r_p^2$. 

To illustrate this point, we employ a simulation study using a dataset generated as follows
$$
y = 1 + f_1(x_{i1}) + f_2(x_{i2}) + f_3(x_{i3}) + f_4(x_{i4}) + f_5(x_{i5}) + \epsilon_i.$$
Here, $y_i$ represents the response variable, and $(x_{i1}, \dots, x_{i5})$ denotes the five predictors, each independently sampled from a uniform distribution in the range $(-1, 1)$.  Errors $\epsilon_i$'s are generated as iid samples from $N(0, 1)$. The sample size is  $n = 500$. 

The ground truth functions are defined as follows,
\begin{equation}
\label{eq:smoothness_md_ground_truth}
    f_1(x) = 8\sin(12x),\quad
    f_2(x) = 3x^3,\quad
    f_3(x) = -2x,\quad
    f_4(x) = 0,\quad
    f_5(x) = 0.
\end{equation}
As we can notice from their definitions,  $f_1$ to $f_5$ have varying level of smoothness. $f_1$ and $f_2$ are nonlinear, $f_3$ is linear, and $f_4$ and $f_5$ are zero functions. We tested GAMSEL and VARD on this dataset and plotted the fitted functions $\hat{f}_1$ to $\hat{f}_5$ in Figure \ref{fig:smoothness_md}. To ensure a fair comparison, for GAMSEL, we set $dim = df = 25$ and used $\gamma = 0.4$. For VARD, we used 25 knots and the natural cubic spline basis. We evaluated both $\lambda_{min}$ and $\lambda_{1se}$ for GAMSEL, as recommended by the authors. In our simulations, we found that $\lambda_{1se}$ performed slightly better than $\lambda_{min}$.  In Figure \ref{fig:smoothness_md}, we present the results of $\lambda_{1se}$ for GAMSEL and the results of $\alpha_{0.15se}$ for VARD.

\begin{figure}[!ht]
\centering
{\includegraphics[width=\textwidth]{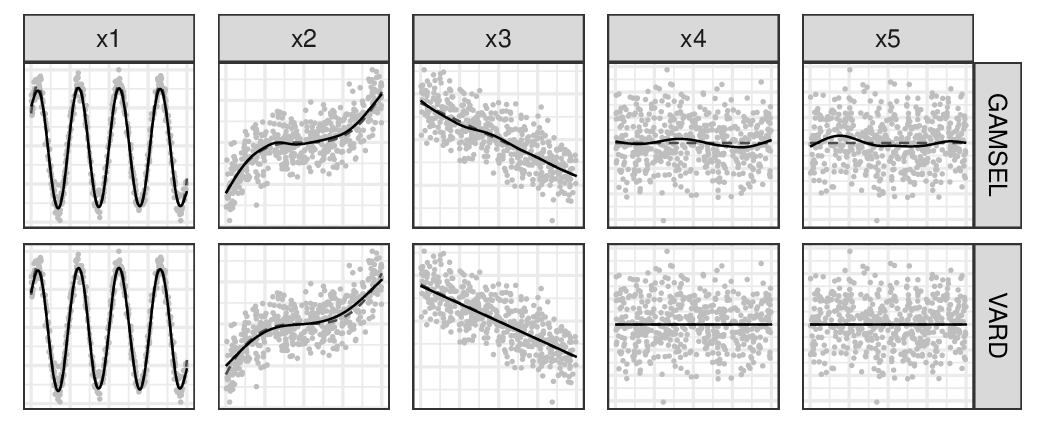}}
\caption[short]{Simulation study on fitting the smoothness for method GAMSEL and VARD in multi-feature case. The solid curves represent the fitted functions $\hat{f_j}$ for each method, while the dashed curves depict the ground truth functions $f_j$ as defined in equation \eqref{eq:smoothness_md_ground_truth}. The grey dots represent the simulated data points, marginally given by $\{(x_{ij},f_j(x_{ij}) + \epsilon_i)\}_{i=1}^n$. VARD, with $\alpha_{0.15se}$ as the tuning parameter, achieves an exact linear fit for $f_3$ and exact zero fits for $f_4$ and $f_5$. However, GAMSEL, with $\lambda_{1se}$, struggles to individually asses the smoothness.} 
\label{fig:smoothness_md}
\end{figure}

From Figure \ref{fig:smoothness_md}, we observe that:
\begin{itemize}
    \item GAMSEL struggles to ascertain the right individual smoothness for each feature despite hyperparameter tuning. It struggles to precisely identify $f_3$ as being purely linear and $f_4$ and $f_5$ as being strictly zero functions.  This difficulty arises from the global penalty effect embedded within GAMSEL's objective. In situations where there are both features with highly oscillatory functions (e.g., $f_1$) and extremely smooth features (e.g., $f_3$, $f_4$, and $f_5$), GAMSEL faces an inherent trade-off. Ultimately, GAMSEL tends to favor the oscillatory characteristics of $f_1$, leading to oscillations in the fitted curves of other features.
    
    \item VARD, on the other hand, excels in effectively tailoring the smoothness for each feature. We can accurately recognize $f_3$ as genuinely linear, and $f_4$ and $f_5$ as precisely zero functions. This is again a testament to its ability to learn each feature's smoothness, parameterized by $r_j^2$ in \eqref{eq:elbo_alter_2}, individually and adaptively.
\end{itemize}

In summary, the studies comparing single-feature and multi-feature cases reveal distinct performance differences among the methods. SPAM lacks the capability to fit smoothness effectively, and GAMSEL encounters challenges in accommodating the smoothness of each feature separately. Whereas VARD demonstrates superiority in accurately capturing the right smoothness for all features.

\subsection{Performance Comparison}
\label{sec:performance_comparison}

In this section, we present the results of our simulation study, conducted under various data generation setups, to compare the performance of VARD with SPAM and GAMSEL in terms of both \textit{estimation and selection accuracy}. As SPAM cannot distinguish between linear and non-linear functions, while VARD and GAMSEL can, we designed two experiments to evaluate their capabilities. In Experiment 1, we assess the selection accuracy of identifying non-zero functions versus zero functions for all three methods. In Experiment 2, we focus on comparing the selection accuracy between GAMSEL and VARD in distinguishing between zero functions, linear functions, and non-linear functions. Additionally, we compare the estimation accuracy across all three methods in the first experiment. For both experiments, we explore several cases of different data-generating processes by varying the following aspects of additive model \eqref{eq:additive_model_regression} in Table \ref{tb:experiment_notation}: 

\begin{table}[ht]
\small
\centering
\caption{Experiment Notations}
\label{tb:experiment_notation}
\begin{tabular}{l|l}
\toprule
$n$: sample size & $p$: number of features \\
\midrule
$s$: number of relevant features & $\sigma^2$: error variance \\
\midrule
$\zeta$: marginal distribution of each feature & $\rho$: correlation between any two features \\
\midrule
$s_n$: number of nonlinear features & $s_l$: number of linear features \\
\bottomrule
\end{tabular}
\end{table}

In Table \ref{tb:experiment_notation}, $s$ is considered in Experiment 1, while  $s_n$ and $s_l$ are considered in Experiment 2. For each relevant feature $x_j$, we assign its true function $f_j$ as one of the following functions multiplied by some non-zero constant:
$$
    \phi_1(x) = 10e^{-4.6x^2}, \quad
    \phi_2(x) = 4\cos(1.7x), \quad
    \phi_3(x) = 5(x + 1.3)^2, \quad
    \phi_4(x) = 6(x + 5).
$$
Notice that $\phi_1, \phi_2, \phi_3$ are nonlinear functions and $\phi_4$ is a linear function. 

For each case, we generated 100 datasets and ran SPAM, GAMSEL, and VARD. We report the averages and standard deviations of each experiment's target metrics. As before, to ensure a fair comparison, we kept the number of basis spline functions the same across all models. Specifically, we set the parameter $p = 10$ for SPAM, $dim = 10$ for GAMSEL, and the number of knots as 10 for VARD. For VARD, we used the natural cubic spline basis. In the case of GAMSEL, we set $df = 6$ and $\gamma = 0.4$.

Regarding SPAM, we did not employ the recommended GCV or Cp criteria for hyperparameter selection, as the authors did not specify how to estimate error variance for computing Cp, and we found that GCV led to suboptimal hyperparameters. Instead, we conducted a 10-fold cross-validation and reported the performance of SPAM using $\lambda_{min}$ and $\lambda_{1se}$. For GAMSEL, we utilized their built-in \texttt{cv.gamsel} function to report the model performance corresponding to $\lambda_{min}$ and $\lambda_{1se}$. For our model, we report the performance corresponding to $\alpha_{min}$ and $\alpha_{0.15se}$.

Finally, we first conducted manual ten-fold cross-validation on the initial dataset. This ensured that the hyperparameter sequences for each model included the minimum or (1 or 0.15)se values. We aimed to confirm that the validation mean squared error (MSE) path showed a U shape and included the desired hyperparameters for all three methods. After confirming this with the first dataset, we used the identified minimum or (1 or 0.15)se hyperparameters for the remaining 99 datasets. This approach ensured consistency across all three methods for all datasets.

\subsubsection{Experiment 1}
In this first experiment, we aim to compare the estimation and selection accuracy of all three methods (SPAM, GAMSEL, and VARD). Selection accuracy in this context refers to the ability to correctly determine whether a feature's function is zero or non-zero. We consider the six data generating cases in Table \ref{tb:experiment_1_setup}, the detailed assignment of non-zero function for each case can be found in Appendix \ref{sec:appendix_experiment_setup}:

\begin{table}[ht]
\centering
\caption{Experiment 1 Setup (Notation defined in Table \ref{tb:experiment_notation})}
\label{tb:experiment_1_setup}
\scalebox{1}{
\begin{tabular}{c|c|c|c|c|c|c}
\toprule
& \textbf{Case 1} & \textbf{Case 2} & \textbf{Case 3} & \textbf{Case 4} & \textbf{Case 5} & \textbf{Case 6} \\
\midrule
$n$ & 500 & 800 & 500 & 1000 & 500 & 500\\ 
$p$ & 10 & 15 & 150 & 1000 & 30 & 30\\
$s$ & 4 & 8 & 3 & 4 & 12 & 12\\
$\sigma^2$ & 1 & 4 & 1 & 1 & 1 & 1\\
$\zeta$ & $\mathcal{U}(-1, 1)$ & $\mathcal{U}(-1, 1)$ & $\mathcal{U}(-1, 1)$ & $\mathcal{U}(-1, 1)$ & $\mathcal{N}(0, 1)$ & $\mathcal{N}(0, 1)$\\
$\rho$ & 0 & 0 & 0 & 0 & 0.3 & 0.7\\
\bottomrule
\end{tabular}
}
\end{table}

To assess estimation accuracy, we employed the in-sample MSE (mean squared error), defined as $\frac{1}{n}\sum_{i=1}^n\sum_{j=1}^p (f_j(x_{ij}) - \hat{f}_j(x_{ij}))^2$, as the metric. To test feature selection accuracy, we treated the relevant features as positives and irrelevant features as negatives, utilizing False Discovery Rate (FDR) and True Positive Rate (TPR) as metrics. For all three methods, the `\textit{min}' hyperparameter is more accurate in estimation accuracy and less accurate in selection accuracy than the `\textit{se}' hyperparameter within each method. Thus, for brevity, when it comes to estimation metric in-sample MSE, we report results of SPAM and GAMSEL with $\lambda_{min}$ and our VARD with $\alpha_{min}$, and when it comes to selection accuracy metrics FDR and TPR, we report results of SPAM and GAMSEL with $\lambda_{1se}$ and our VARD with $\alpha_{0.15se}$. For each case, we provided the averages and standard deviations of these three metrics over 100 generated datasets in tables \ref{tb:estimation_selection_1_4}, \ref{tb:estimation_selection_5_6}. We summarize our observations as follows 
\begin{itemize}
    \item VARD demonstrates superior estimation accuracy across all six cases, consistently achieving the lowest in-sample MSE.
    \item VARD excels in distinguishing between zero and non-zero functions. Along with SPAM, it is among the top methods for feature selection, significantly outperforming GAMSEL and showing comparable performance to each other.
    \item VARD maintains strong performance even in the presence of feature correlations (Table \ref{tb:estimation_selection_5_6}). In contrast, GAMSEL and SPAM both struggle with correlated features.
\end{itemize}

In this experiment, across several situations, we can see that our method VARD has an advantage in both estimation and selection accuracy over SPAM and GAMSEL. In practice, although  VARD(min) achieves slightly better estimation error, we still recommend VARD(0.15se) due to its high accuracy in feature selection.

% Table for Cases 1-4
\begin{table}[H]
\centering
\caption{Estimation and Selection Accuracy Comparison between SPAM, GAMSEL, and VARD (Cases 1-4)}
\label{tb:estimation_selection_1_4}
\scalebox{0.95}{
\begin{tabular}{c|c|c|c|c|c}
\toprule
\textbf{Metric} & \textbf{Method} & \textbf{Case 1} & \textbf{Case 2} & \textbf{Case 3} & \textbf{Case 4} \\
\midrule
\multirow{3}{*}{\hfil \textbf{MSE}} 
& SPAM & 0.164$\pm$0.03 & 0.706$\pm$0.09 & 0.204$\pm$0.04 & 0.146$\pm$0.02 \\
& GAMSEL & 0.102$\pm$0.02 & 3.37$\pm$0.34 & 0.085$\pm$0.02 & 0.347$\pm$0.03 \\
& VARD & \textbf{0.040$\pm$0.01} & \textbf{0.453$\pm$0.07} & \textbf{0.029$\pm$0.01} & \textbf{0.026$\pm$0.01} \\
\midrule
\multirow{3}{*}{\hfil \textbf{FDR}} 
& SPAM & 0$\pm$0 & 0$\pm$0 & 0$\pm$0 & 0$\pm$0 \\
& GAMSEL & 0.467$\pm$0.08 & 0.466$\pm$0.01 & 0$\pm$0 & 0.607$\pm$0.08 \\
& VARD & 0$\pm$0 & 0.018$\pm$0.04 & 0$\pm$0 & 0$\pm$0 \\
\midrule
\textbf{TPR} & \textbf{ALL} & 1$\pm$0 & 1$\pm$0 & 1$\pm$0 & 1$\pm$0 \\
\bottomrule
\end{tabular}
}
\end{table}

% Table for Cases 5-6
\begin{table}[H]
\centering
\caption{Estimation and Selection Accuracy Comparison between SPAM, GAMSEL, and VARD (Cases 5-6)}
\label{tb:estimation_selection_5_6}
\scalebox{0.95}{
\begin{tabular}{c|c|c|c}
\toprule
\textbf{Metric} & \textbf{Method} & \textbf{Case 5} & \textbf{Case 6} \\
\midrule
\multirow{3}{*}{\hfil \textbf{MSE}} 
& SPAM & 0.613$\pm$0.07 & 1.04$\pm$0.17 \\
& GAMSEL & 4.21$\pm$0.47 & 4.16$\pm$0.55 \\
& VARD & \textbf{0.465$\pm$0.08} & \textbf{0.449$\pm$0.09} \\
\midrule
\multirow{3}{*}{\hfil \textbf{FDR}} 
& SPAM & 0.007$\pm$0.03 & 0.108$\pm$0.07 \\
& GAMSEL & 0.587$\pm$0.01 & 0.588$\pm$0.01 \\
& VARD & \textbf{0$\pm$0} & \textbf{0.002$\pm$0.01} \\
\midrule
\textbf{TPR} & {ALL} & 1$\pm$0 & 1$\pm$0 \\
\bottomrule
\end{tabular}
}
\end{table}

\subsubsection{Experiment 2}
In this experiment, we run simulations to compare the power of differentiating nonlinear, linear, and zero functions between GAMSEL and VARD. SPAM is not in this experiment because it can't distinguish linear function vs nonlinear function. We consider the following 5 data generating cases in Table \ref{tb:experiment_2_setup}, the detailed linear and nonlinear function assignments can be found in Appendix \ref{sec:appendix_experiment_setup}:

\begin{table}[ht]
\centering
\caption{Experiment 2 Setup (Notation defined in Table \ref{tb:experiment_notation})}
\label{tb:experiment_2_setup}
\scalebox{1}{
\begin{tabular}{c|c|c|c|c|c}
\toprule
 & \textbf{Case 1} & \textbf{Case 2} & \textbf{Case 3} & \textbf{Case 4} & \textbf{Case 5} \\
\midrule
$n$ & 600 & 2000 & 1000 & 600 & 500 \\ 
$p$ & 18 & 100 & 1200 & 30 & 30 \\
$s_n$ & 6 & 25 & 3 & 10 & 10 \\
$s_l$ & 6 & 25 & 3 & 10 & 10 \\
$\zeta$ & $\mathcal{U}(-1, 1)$ & $\mathcal{U}(-1, 1)$ & $\mathcal{U}(-1, 1)$ & $\mathcal{N}(0, 1)$ & $\mathcal{N}(0, 1)$ \\
$\rho$ & 0 & 0 & 0 & 0.3 & 0.7 \\
\bottomrule
\end{tabular}
}
\end{table}

We focus on the resulting $3 \times 3$ confusion matrix. We report the averages and standard deviations of each entry in the confusion matrix over 100 generated data. The result can be found in Table \ref{tb:nonlinear_linear_zero}. We did not report the result of VARD(min), which has slightly worse selection accuracy than VARD(0.15se),  for brevity purposes. We report both GAMSEL(min), GAMSEL(1se) to illustrate that none of GAMSEL's best choice of hyperparameter can be better than ours.

\begin{table}[ht]
\centering
\caption{Distinguishing Nonlinear, Linear, and Zero Functions (GAMSEL vs VARD)}
\label{tb:nonlinear_linear_zero}
\scalebox{0.68}{
\begin{tabular}{@{\hskip 2pt} c|c|ccc|ccc|ccc @{\hskip 2pt}}
\toprule
\multicolumn{2}{c|}{\multirow{2}{*}{\diagbox{Truth}{Model}}} & \multicolumn{3}{c|}{GAMSEL (min)} & \multicolumn{3}{c|}{GAMSEL (1se)} & \multicolumn{3}{c}{VARD (0.15se)} \\
\cline{3-11}
\multicolumn{2}{c|}{} & nonlinear & linear & zero & nonlinear & linear & zero & nonlinear & linear & zero \\
\midrule
\multirow{3}{*}{Case 1} & nonlinear & 6.0$\pm$0.0 & - & - & 6.0$\pm$0.0 & - & - & \textbf{6.0$\pm$0.0} & - & - \\
& linear & 5.6$\pm$0.6 & 0.4$\pm$0.6 & - & 6.0$\pm$0.0 & - & - & 0.0$\pm$0.1 & \textbf{6.0$\pm$0.1} & - \\
& zero & 5.6$\pm$0.5 & 0.3$\pm$0.5 & 0.1$\pm$0.2 & 5.1$\pm$0.9 & - & 0.9$\pm$0.9 & 0.0$\pm$0.1 & 0.0$\pm$0.1 & \textbf{6.0$\pm$0.2} \\
\midrule
\multirow{3}{*}{Case 2} & nonlinear & 25.0$\pm$0.0 & - & - & 25.0$\pm$0.0 & - & - & \textbf{25.0$\pm$0.0} & - & - \\
& linear & 24.9$\pm$0.3 & 0.1$\pm$0.3 & - & 25.0$\pm$0.0 & - & - & 0.0$\pm$0.1 & \textbf{25.0$\pm$0.1} & - \\
& zero & 49.7$\pm$0.5 & 0.3$\pm$0.5 & - & 49.7$\pm$0.6 & - & 0.3$\pm$0.6 & 0.0$\pm$0.2 & 0.1$\pm$0.3 & \textbf{49.8$\pm$0.4} \\
\midrule
\multirow{3}{*}{Case 3} & nonlinear & 3.0$\pm$0.0 & - & - & 3.0$\pm$0.0 & - & - & \textbf{3.0$\pm$0.0} & - & - \\
& linear & 0.0$\pm$0.2 & 3.0$\pm$0.2 & - & 3.0$\pm$0.0 & - & - & - & \textbf{3.0$\pm$0.0} & - \\
& zero & 5.9$\pm$2.6 & 212.6$\pm$12.0 & 975.6$\pm$12.0 & 3.1$\pm$1.7 & - & 1191.0$\pm$1.7 & 0.3$\pm$0.5 & 0.9$\pm$1.0 & \textbf{1193.0$\pm$1.1} \\
\midrule
\multirow{3}{*}{Case 4} & nonlinear & 10.0$\pm$0.0 & - & - & 10.0$\pm$0.0 & - & - & \textbf{10.0$\pm$0.0} & - & - \\
& linear & 9.3$\pm$0.9 & 0.7$\pm$0.9 & - & 10.0$\pm$0.0 & - & - & 0.0$\pm$0.2 & \textbf{10.0$\pm$0.2} & - \\
& zero & 9.3$\pm$0.8 & 0.6$\pm$0.7 & 0.2$\pm$0.4 & 6.3$\pm$1.6 & - & 3.7$\pm$1.6 & 0.0$\pm$0.1 & 0.0$\pm$0.1 & \textbf{10.0$\pm$0.2} \\
\midrule
\multirow{3}{*}{Case 5} & nonlinear & 10.0$\pm$0.0 & - & - & 10.0$\pm$0.0 & - & - & \textbf{10.0$\pm$0.0} & - & - \\
& linear & 9.3$\pm$0.8 & 0.7$\pm$0.8 & - & 10.0$\pm$0.0 & - & - & 0.0$\pm$0.1 & \textbf{10.0$\pm$0.1} & - \\
& zero & 9.1$\pm$1.0 & 0.5$\pm$0.8 & 0.4$\pm$0.6 & 6.0$\pm$1.4 & - & 4.0$\pm$1.4 & 0.0$\pm$0.1 & - & \textbf{10.0$\pm$0.1} \\
\bottomrule
\end{tabular}
}
\end{table}

Table \ref{tb:nonlinear_linear_zero} reveals several key trends across all cases. GAMSEL(min) often misidentifies zero functions as non-zero and frequently struggles with linear functions. In comparison, GAMSEL(1se) performs better at recognizing whether a function is zero or not but struggles to distinguish between non-zero functions that are linear or nonlinear. In contrast, our approach, VARD(0.15se), excels in accurately classifying functions as nonlinear, linear, or zero in all five cases.

As discussed in Section \ref{sec:Multiple_Feature_Case}, the challenges faced by GAMSEL, particularly in complex scenarios with multiple features. GAMSEL applies a uniform penalty to the smoothness of all features, thus lacking the ability to adapt to the individual smoothness of each feature. Conversely, our method overcomes this limitation, allowing it to perform effectively in such scenarios.

\section{Conclusion}
In this paaper, we introduced a Variational Automatic Relevance Determination (VARD) framework for additive models. This framework automatically detects the smoothness of each feature while achieving exact sparsity, thereby classifying a feature's influence on the response variable as zero, linear, or nonlinear. To optimize the framework's objective function, we designed a coordinate descent algorithm that naturally attains exact sparsity at each iteration. Through empirical study, we established the algorithm's superiority over existing sparsity-inducing algorithms like SPAM and GAMSEL. Specifically, our method outperforms these benchmarks in estimation accuracy, selection accuracy, and power of distinguishing between zero, linear, and nonlinear feature contributions. VARD further boasts the simplicity of having to tune a single hyperparameter, much like the lasso method. Overall, our results demonstrate that our framework is a robust and efficient methodology for feature selection and smoothness learning simultaneously. 

To guide our future research efforts, we aim to bridge the gap between the theoretical advancements presented in this paper and their practical application in real-world scenarios. Specifically, we intend to explore scenarios where the number of features significantly exceeds that of observations. While this exploration is beyond the scope of our current work due to its substantial computational demands, we acknowledge the importance of addressing this challenge.

In our algorithm, despite its effectiveness in managing individual smoothing and selection, we still rely on grid search for optimization in some parts. This reliance creates a bottleneck when dealing with cases involving a large number of features. In future iterations of our research, we aspire to develop more efficient optimization techniques to overcome this limitation.

Additionally, we see an opportunity to enhance our method by providing theoretical guarantees regarding the accuracy of feature selection. The innovative variational approach we've introduced to regularize and optimize additive models stands out for its novelty, efficiency, and remarkable accuracy. We are confident that by extending our theoretical framework and conducting more extensive practical experiments, our approach has the potential to achieve `state-of-the-art' status in the realm of additive models.
\newpage
\bibliographystyle{apalike}
\bibliography{ref2}

\newpage

\appendixtitleon
\appendixtitletocon

\setcounter{page}{1}
\begin{center}
{\Large Supplementary Materials for Bayesian Smoothing and Feature Selection using
Variational Automatic Relevance Determination
}
\end{center}

\begin{appendices}

\section{}
\label{sec:appendix_regularization_review}
In this part of the appendix, we provide a detailed review on the four regularization methods mentioned in Section \ref{sec:regularization_methods}: smoothing spline, COSSO, SPAM and GAMSEL. The objective functions of all these methods  take the following form:
\begin{equation*}
%\label{eq:regularization_obj_appendix}
\text{Objective}(f) = \text{RSS}(f) + \text{Penalty}(f)
\end{equation*}
where $\text{RSS}(f)$ represents the residual sum of squares, given by $\|\mathbf{y} - \sum_{j=1}^p f_j(\bm{x}_j)\|^2$ and $\text{Penalty}(f)$ is composed of ridge or group lasso: 
\begin{equation*}
%\label{eq:pen_form_appendix}
    Penalty(f) =  \underbrace{\sum_{j=1}^p \lambda_jJ_R(f_j)}_\text{Ridge} + \underbrace{\lambda\sum_{j=1}^pJ_L(f_j)}_\text{Group Lasso}.
\end{equation*}

\paragraph{Smoothing spline} In the context of the additive model, the classical smoothing spline objective consists of the residual sum of squares and a penalty term on the roughness of each $f_j$, quantified as the integration of the square of the curve's second-order derivative.
\begin{equation}
\label{eq:smoothing_spline_obj_appendix} 
\underset{\{f_j\}_{j=1}^p}{\argmin} \|\bm{y} - \sum_{j=1}^pf_j(\bm{x}_j)\|_2^2 + \underbrace{\sum_{j=1}^p\lambda_j \underbrace{\int f_j^{''}(t)^2dt}_{J_R(f_j)}}_{Ridge}
\end{equation}
Here $\{\lambda_j\}_1^p$ is a set of non-negative penalty parameters that control the smoothness of each $f_j$.

To understand why the penalty part takes on a ridge form, it can be demonstrated that the optimal $\hat{f}_j$ must be a natural cubic spline with $n$ knots at each data point of $\bm{x}_j$ (Wahba, 1990). In practice, for computational efficiency, it is also common to directly set $f_j$ as a natural cubic spline of the form \eqref{eq:ncs_basis_representation} with fewer knots ($d_j + 2$) than $n$:
\begin{equation}
\label{eq:ncs_basis_representation}
  f_j(x) = \beta_0 + x \beta_{j0} + \sum_{k=1}^{d_j} h_{jk}(x)\beta_{jk}.
\end{equation}
Here, $\{h_{jk}\}_{k=1}^{d_j}$  represents the set of $d_j$ non-linear natural cubic spline basis functions, $\{\beta_{jk}\}_{k=1}^{d_j}$ are the coefficients corresponding to these non-linear basis functions, $\beta_0$ is the intercept term, and $\beta_{j0}$  is the coefficient for the linear term of the $j$-th feature. By consolidating the intercept terms from all $f_j$ into a single term, $\beta_0$, problem \eqref{eq:smoothing_spline_obj_appendix} is transformed into a ridge-like problem as follows: 
\begin{equation}
\label{eq:smoothing_spline_matrix_obj}
    \underset{\beta_0, \{\beta_{j1}\}_1^p,\{\bm{\beta}_j\}_1^p}{\argmin} \|\bm{y} - \beta_0 - \sum_{j=1}^p \bm{x}_j\beta_{j0} -  \sum_{j=1}^p H_j\bm{\beta}_j\|_2^2 + \sum_{j=1}^p \lambda_j \underbrace{\bm{\beta}_j^T S_j \bm{\beta}_j}_{J_R(f_j)}.
\end{equation}
In this equation, $\bm{\beta}_j = (\beta_{j1}, \cdots, \beta_{jd_j})$, $H_j = [h_{jk}(x_{ij})]_{i=1, k=1}^{n, d_j}$ is the basis matrix for the $j$-th feature (without intercept and linear term), and $S_j = [\int h^{''}_{jk_1}(x)h^{''}_{jk_2}(x)\mathrm{d}x]_{k_1=1,k_2=1}^{d_j,d_j}$ is the smoothing matrix for the $j$-th feature. %%Given a set of smoothing parameters $\{\lambda_j\}_1^p$, the solution for \eqref{eq:smoothing_spline_matrix_obj} is in closed form.

While the smoothing spline method laid the foundation for representing the roughness of each $f_j$ as  $\int f_j^{''}(t)^2dt$ in the additive model, a practical challenge remained in how to effectively tune the numerous smoothness parameters $\lambda_j$, particularly when dealing with high dimensions. This challenge was elegantly addressed by \cite{woodbook}, who proposed a method for estimating each $f_j$'s smoothness parameter $\lambda_j$ by optimizing a Generalized Cross Validation or Restricted (or Residual) Maximum Likelihood objective within the framework of smoothing splines. Additionally, they introduced the \texttt{mgcv} R package, which is considered the state-of-the-art tool for controlling smoothness. Despite the significant success of \texttt{mgcv} in regulating smoothness using smoothing splines, this approach encounters difficulties in feature selection for two primary reasons: 1) it imposes no penalty on linear terms, and 2) the solution cannot reach exact sparsity. The first issue can be readily addressed by introducing additional ridge-like penalty terms for linear coefficients into objective \eqref{eq:smoothing_spline_matrix_obj}, as recommended by \cite{marra2011}. However, the second issue lacks an easy solution due to the inherent nature of ridge penalties: regardless of how large $\lambda_j$ we use, it will never force the corresponding coefficient to exactly zero.

%%In the literature, feature selection in additive models is often accomplished through group lasso regularization \citep{yuan2006model}. The forms of group lasso penalty are motivated by regularizing different functional norms of $f_j$ in various works. Notable examples include COSSO \citep{lin2006}, SPAM \citep{ravikumar2009}, and GAMSEL \citep{alex2015}. The objective functions in these works all involve a group lasso penalty term with different Mahalanobis norms, respectively. 

\paragraph{COSSO} \cite{lin2006} proposed the COSSO method by penalizing the norm of a curve defined based on all derivatives of the curve from zero to $\ell$-th order. In the context of the additive model, the objective function for COSSO is given by:
\begin{equation}
\label{eq:COSSO_obj}
\underset{\{f_j\}_{j=1}^p}{\argmin}  \|\bm{y} - \sum_{j=1}^pf_j(\bm{x}_j)\|_2^2 + \underbrace{\lambda\sum_{j=1}^p J_L(f_j)}_{\text{Group Lasso}}
\end{equation}
where $J_L(f_j)$ is the norm of $f_j$ defined as,
$$J_L(f_j) := \sqrt{\sum_{\upsilon = 0}^{\ell - 1}\{\int f_j^{(\upsilon)}(t)dt\}^2 + \int\{f_j^{(\ell)}(t)\}^2dt}$$
and $\lambda$ is the non-negative group lasso penalty parameter. This norm measures the roughness of $f_j$ by taking every derivative of $f_j$ up to $\ell$-th order into consideration. 
%To see why COSSO is a group lasso method after basis expansion, as suggested by the author, please consider the terms: 
%\begin{equation} \label{eq:basis_expansion_appendix}
%    f_j(x) = \sum_{k=1}^d h_{jk}(x)\beta_{jk}.
%\end{equation}
The COSSO objective takes a group lasso form as follows,
\begin{equation*}
%\label{eq:COSSO_obj_matrix}
\underset{\{\bm{\beta}_j\}_{j=1}^p}{\argmin} \|\bm{y} - \sum_{j=1}^p H_j\bm{\beta}_j\|_2^2 + \lambda\sum_{j=1}^p \sqrt{\bm{\beta}_j^T S_j\bm{\beta}_j}
\end{equation*}
where $H_j = [h_{jk}(x_{ij})]_{i=1, k=1}^{n, d}$ is the basis matrix for the $j$-th feature, $\bm{\beta}_j = (\beta_{j1}, \cdots, \beta_{jk})$ is the coefficient vector for the $j$-th feature, and $S_j = [\sum_{\upsilon=0}^{\ell -1}\int h_{jk_1}^{(\upsilon)}(x)\mathrm{d}x\int h_{jk_2}^{(\upsilon)}(x)\mathrm{d}x + \int h^{(\ell)}_{jk_1}(x) h^{(\ell)}_{jk_2}(x)\mathrm{d}x]_{k_1=1,k_2=1}^{d_j,d_j}$ is the COSSO smoothing matrix for the $j$-th feature. Although COSSO achieves precise sparsity using group lasso and penalizes the roughness of each $f_j$ through its norm $J(f_j)$, it differs from smoothing splines in that the smoothness of every $f_j$ is regulated by a single, universal penalty parameter, $\lambda$. This characteristic causes the method to struggle to adapt to varying smoothness levels for different features.

\paragraph{SPAM} \cite{ravikumar2007spam} and \cite{ravikumar2009} proposed the SPAM method, wherein the objective function takes a similar form to that of COSSO's objective; see equation \eqref{eq:COSSO_obj}. Here, however,  the norm $J(f_j)$ is defined as the square root of the expectation of the square of $f_j$:
$$J(f_j):= \sqrt{E(f_j^2(x_j))}$$
In practice, the expectation is replaced by its empirical counterpart:
$$J(f_j):= \sqrt{\frac{1}{n}\sum_{i=1}^n f_j^2(x_{ij})}$$
It is important to note that the SPAM norm is not a measure of the roughness of $f_j$ but rather a simple assessment of the magnitude of $f_j$. Again, by expressing $f_j$ in terms of basis functions as shown in equation \eqref{eq:ncs_basis_representation}, the SPAM objective adopts a group lasso form:
\begin{equation}
\label{eq:SPAM_matrix_obj} 
\underset{\{\bm{\beta}_j\}_{j=1}^p}{\argmin} \|\bm{y} - \sum_{j=1}^p H_j\bm{\beta}_j\|_2^2 + \lambda\sum_{j=1}^p \sqrt{\bm{\beta}_j^T\frac{H_j^T H_j}{n}\bm{\beta}_j}.
\end{equation}
Like COSSO, SPAM achieves exact sparsity due to its group lasso objective. Although successful experiments with SPAM for feature selection have been reported, it lacks smoothing capabilities as it does not penalize the roughness of each $f_j$ (its definition of $f_j$ does not measure roughness). In our experiments, we observed that a larger group lasso parameter $\lambda$ in SPAM does not encourage a smoother $\hat{f}_j$ (see Figure \ref{fig:smoothness_1d}). Instead, a larger $\lambda$ merely compresses $\hat{f}_j$ along the y-axis.

\paragraph{GAMSEL} \cite{alex2015} proposed the GAMSEL method with an objective derived from the matrix representation of smoothing spline objective in Equation \eqref{eq:smoothing_spline_matrix_obj} by adding a group lasso penalty term as follows:
\begin{eqnarray}
  &&  \underset{\beta_0, \{\beta_{j0}\}_1^p,\{\bm{\beta}_j\}_1^p}{\argmin}    \|\bm{y} - \beta_0 - \sum_{j=1}^p \bm{x}_j\beta_{j0} -  \sum_{j=1}^p H_j\bm{\beta}_j\|_2^2 + \nonumber \\
   && \quad \qquad \underbrace{\lambda\sum_{j=1}^p\underbrace{\left\{\gamma|\beta_{j0}| + (1-\gamma)\sqrt{\bm{\beta}_j^TS_j^*\bm{\beta}_j}\right\}}_{J_L(f_j)}}_{\text{Group Lasso}} + \underbrace{\frac{1}{2}\sum_{j=1}^p \lambda_j \underbrace{\bm{\beta}_j^T S_j \bm{\beta}_j}_{J_R(f_j)}}_{\text{Ridge}}   \label{eq:gamsel_obj}
\end{eqnarray}
where $\lambda$ is the sparsity-promoting group lasso penalty parameter that can be tuned using cross-validation, $\gamma \in (0, 1)$ is a pre-specified hyperparameter to balance the lasso penalty between linear and nonlinear terms, and $\lambda_1, \cdots, \lambda_p$ are $p$ pre-specified smoothing parameters. To be precise, the definitions of $H_j$, $\bm{\theta}_j$, and $S_j$ in Equation \eqref{eq:gamsel_obj} differ slightly from those in Equation \eqref{eq:smoothing_spline_matrix_obj}. In Equation \eqref{eq:gamsel_obj}, $H_j$ and $\bm{\beta}_j$ include linear terms, and $S_j$ has an additional all-zero row and column corresponding to the linear term. $S_j^*$ is identical to $S_j$, except that the diagonal position corresponding to the linear term is replaced from zero to one. For further details, readers can refer to GAMSEL \cite{alex2015}. 

Although GAMSEL achieves exact sparsity by incorporating a group lasso penalty term, it differs from the smoothing spline method, where smoothing parameters $\lambda_1, \cdots, \lambda_p$ can be learned directly by optimizing a GCV or REML function. In GAMSEL, these smoothness parameters must be pre-specified manually, which is impractical in high-dimensional settings. Moreover, since the group lasso penalty term in GAMSEL originates from the smoothing spline ridge-like penalty term, which corresponds to a roughness measurement for each $f_j$, the group lasso parameter $\lambda$ in GAMSEL also affects the smoothness of every $f_j$ universally. This makes it even more challenging to determine the appropriate smoothness for each $f_j$ in practice.

\section{}
\label{sec:appendix_standardization}
In this part of appendix, for centered feature $\bm{x}_j = (x_{1j}, \cdots, x_{nj})^T$, with basis expansion of form:
$$f_j(x) = \beta_{j0}x + \sum_{k=1}^{d_j}\beta_{jk}h_{jk}(x)
$$
where the nonlinear basis $\{h_{jk}\}_{k=1}^{d_j}$ is centered (basis basis matrix $H_j:= [h_{jk}(x_{ij})]_{i=1, k=1}^{n, d_j}$ has column mean zero), we provide a standardization procedure on nonlinear basis $\{h_{jk}\}_{k=1}^{d_j}$, such that after standardization, $\{h_{jk}\}_{k=1}^{d_j}$ satisfies the following properties:
\begin{enumerate}
    \item The $j$-th feature's nonlinear basis matrix $H_j := [h_{jk}(x_{ij})]_{i=1, k=1}^{n, d_j}$ is orthogonal to itself and orthogonal to $j$-th feature's linear component $\bm{x}_j$.
    \item The $j$-th feature's smoothing matrix $S_j := [\int h^{''}_{jk_1}(x)h^{''}_{jk_2}(x)\mathrm{d}x]_{k_1=1,k_2=1}^{d_j,d_j}$ for nonlinear basis equals to identity matrix $I_{d_j}$.
\end{enumerate}
and the standardization procedure won't change the functional space the basis functions represent: $\{f_j: f_j(x) = \beta_{j0}x + \sum_{k=1}^{d_j}\beta_{jk}h_{jk}(x)\}$ stays the same before and after standardization.

In the following, we denote $h_{jk}(\bm{x}_j) = (h_{jk}(x_{1j}), \cdots, h_{jk}(x_{nj}))^T$ as the $k$-th column of basis matrix $H_j := [h_{jk}(x_{ij})]_{i=1, k=1}^{n, d_j}$. The standardization procedure is listed as follows:
\begin{procedure}
\caption{Nonlinear Basis Standardization () for $f_j(x) = \beta_{j0}x + \sum_{k=1}^{d_j}\beta_{jk}h_{jk}(x)$}
Input centered $\bm{x}_j$ and centered nonlinear basis $\{h_{jk}\}_{k=1}^{d_j}$\;
    Calculate $H_j = [h_{jk}(x_{ij})]_{i=1, k=1}^{n, d_j}$\;
    \For{$k$ in $1,\dots, d_j$}{
    Update basis $h_{jk}(x) = h_{jk}(x) - \frac{<\bm{x}_j, h_{jk}(\bm{x}_j)>}{\|\bm{x}_j\|_2^2} x$\;
    }
    Calculate $S_j = [\int h^{''}_{jk_1}(x)h^{''}_{jk_2}(x)\mathrm{d}x]_{k_1=1,k_2=1}^{d_j,d_j}$\;
    Update basis $(h_{j1}, \cdots, h_{jd_j})^T = S_j^{-\frac{1}{2}}(h_{j1}, \cdots, h_{jd_j})^T$\;
    Recalculate $H_j = [h_{jk}(x_{ij})]_{i=1, k=1}^{n, d_j}$\;
    Eigenvalue decomposition on $H_j^TH_j = U_jV_jU_j^T$\;
    Update basis $(h_{j1}, \cdots, h_{jd_j})^T = U_j^T(h_{j1}, \cdots, h_{jd_j})^T$\;
\end{procedure}

In the following, we provide the proof that after standardization, the nonlinear basis $\{h_{jk}\}_{k=1}^{d_j}$ satisfies those two properties listed above and standardization procedure won't change functional space $\{f_j: f_j(x) = \beta_{j0}x + \sum_{k=1}^{d_j}\beta_{jk}h_{jk}(x)\}$. We summarize it as a proposition as follows:
\begin{prop}
    Standardization procedure 1 won't change functional space $\{f_j: f_j(x) = \beta_{j0}x + \sum_{k=1}^{d_j}\beta_{jk}h_{jk}(x)\}$, and after standardization we have:
    \begin{itemize}
        \item The $j$-th feature's nonlinear basis matrix $H_j := [h_{jk}(x_{ij})]_{i=1, k=1}^{n, d_j}$ is orthogonal to itself and orthogonal to $j$-th feature's linear component $\bm{x}_j$.
        \item The $j$-th feature's smoothing matrix $S_j := [\int h^{''}_{jk_1}(x)h^{''}_{jk_2}(x)\mathrm{d}x]_{k_1=1,k_2=1}^{d_j,d_j}$ for nonlinear basis equals to identity matrix $I_{d_j}$.
    \end{itemize}
\end{prop}
\begin{proof}
To distinguish different nonlinear basis $\{h_{jk}\}_{k=1}^{d_j}$, basis matrix $H_j := [h_{jk}(x_{ij})]_{i=1, k=1}^{n, d_j}$ and smoothing matrix $S_j := [\int h^{''}_{jk_1}(x)h^{''}_{jk_2}(x)\mathrm{d}x]_{k_1=1,k_2=1}^{d_j,d_j}$ at each step of update, we rewrite procedure 1 as follows for the convenience of proofing:
\begin{enumerate}
    \item $H_j = [h_{jk}(x_{ij})]_{i=1, k=1}^{n, d_j}$
    \item $\dot{h}_{jk}(x) = h_{jk}(x) - \frac{<\bm{x}_j, h_{jk}(\bm{x}_j)>}{\|\bm{x}_j\|_2^2} x$, $k = 1,\cdots, d_j$
    \item $\dot{S}_j = [\int \dot{h}^{''}_{jk_1}(x)\dot{h}^{''}_{jk_2}(x)\mathrm{d}x]_{k_1=1,k_2=1}^{d_j,d_j}$
    \item $(\ddot{h}_{j1}, \cdots, \ddot{h}_{jd_j})^T = \dot{S}_j^{-\frac{1}{2}}(\dot{h}_{j1}, \cdots, \dot{h}_{jd_j})^T$
    \item $\ddot{H}_j = [\ddot{h}_{jk}(x_{ij})]_{i=1, k=1}^{n, d_j}$
    \item Eigenvalue decomposition $\ddot{H}_j^T\ddot{H}_j = U_jV_jU_j^T$
    \item $(\tilde{h}_{j1}, \cdots, \tilde{h}_{jd_j})^T = U_j^T(\ddot{h}_{j1}, \cdots, \ddot{h}_{jd_j})^T$
\end{enumerate}

In this representation, $\{h_{jk}\}_{k=1}^{d_j}$ is the nonlinear basis before standardization, and $\{\tilde{h}_{jk}\}_{k=1}^{d_j}$ is the nonlinear basis after standardization. 

Denote after standardization basis matrix as $\tilde{H}_j := [\tilde{h}_{jk}(x_{ij})]_{i=1, k=1}^{n, d_j}$ and after standardization smoothness matrix as $\tilde{S}_j := [\int \tilde{h}^{''}_{jk_1}(x)\tilde{h}^{''}_{jk_2}(x)\mathrm{d}x]_{k_1=1,k_2=1}^{d_j,d_j}$, suffice to prove:
\begin{enumerate}[label=(\roman*)]
    \item $\{f_j: f_j(x) = \beta_{j0}x + \sum_{k=1}^{d_j}\beta_{jk}h_{jk}(x)\} = \{f_j: f_j(x) = \beta_{j0}x + \sum_{k=1}^{d_j}\beta_{jk}\tilde{h}_{jk}(x)\}$.
    \item $\bm{x}_j^T\tilde{H}_j = \mathbf{0}$.
    \item $\tilde{H}_j^T\tilde{H}_j$ is a diagonal matrix.
    \item $\tilde{S}_j = I_{d_j}$.
\end{enumerate}

To show (i), since $\dot{h}_{jk}(x) = h_{jk}(x) - \frac{<\bm{x}_j, h_{jk}(\bm{x}_j)>}{\|\bm{x}_j\|_2^2} x$, $k = 1,\cdots, d_j$, we have
\begin{equation}
\label{eq:pro_proof_1}
    Span(x, \dot{h}_{j1}(x), \cdots, \dot{h}_{jd_j}(x)) = Span(x, h_{j1}(x), \cdots, h_{jd_j}(x))
\end{equation}
And since $(\tilde{h}_{j1}, \cdots, \tilde{h}_{jd_j})^T = U_j^T\dot{S}_j^{-\frac{1}{2}}(\dot{h}_{j1}, \cdots, \dot{h}_{jd_j})^T$, we have
\begin{equation}
\label{eq:pro_proof_2}
    Span(x, \tilde{h}_{j1}(x), \cdots, \tilde{h}_{jd_j}(x)) = Span(x, \dot{h}_{j1}(x), \cdots, \dot{h}_{jd_j}(x))
\end{equation}
Combine Equation \eqref{eq:pro_proof_1} and \eqref{eq:pro_proof_2} we get (i).

To show (ii), since $\dot{h}_{jk}(x) = h_{jk}(x) - \frac{<\bm{x}_j, h_{jk}(\bm{x}_j)>}{\|\bm{x}_j\|_2^2} x$, $k = 1,\cdots, d_j$, we have
$$\dot{h}_{jk}(\bm{x}_j) = h_{jk}(\bm{x}_j) - \frac{<\bm{x}_j, h_{jk}(\bm{x}_j)>}
{\|\bm{x}_j\|_2^2} \bm{x}_j, k = 1,\cdots, d_j$$
Here $\dot{h}_{jk}(\bm{x}_j)$ is the $k$-th column of $\dot{H}_j := [\dot{h}_{jk}(x_{ij})]_{i=1, k=1}^{n, d_j}$ and $h_{jk}(\bm{x}_j)$ is the $k$-th column of $H_j$, thus
$$<\dot{h}_{jk}(\bm{x}_j), \bm{x}_j> = <h_{jk}(\bm{x}_j), \bm{x}_j> - \frac{<h_{jk}(\bm{x}_j), \bm{x}_j>}{\|\bm{x}_j\|_2^2}\cdot<\bm{x}_j,\bm{x}_j> = 0$$
This implies that $\bm{x}_j^T\dot{H}_j = \mathbf{0}$.
Since $\tilde{H}_j = \ddot{H}_jU_j = \dot{H}_j\dot{S}_j^{-\frac{1}{2}}U_j$, we have (ii) proved as follows:
$$\bm{x}_j^T\tilde{H}_j = \bm{x}_j^T\dot{H}_j\cdot{S}_j^{-\frac{1}{2}}U_j = \mathbf{0}$$

To show (iii), since $\tilde{H}_j = \ddot{H}_jU_j$, we have
$$\tilde{H}_j^T\tilde{H}_j = U_j^T\ddot{H}_j^T\ddot{H}_jU_j = U_j^TU_jV_jU_j^TU_j = V_j$$
Since $V_j$ is the diagonal matrix from eigenvalue decomposition, (iii) proved.

To show (iv), denote functional vector $\tilde{\bm{h}}^{''}_j(x) := (\tilde{h}^{''}_{j1}(x), \cdots, \tilde{h}^{''}_{jd_j}(x))^T$ and\\$\dot{\bm{h}}^{''}_j(x) := (\dot{h}^{''}_{j1}(x), \cdots, \dot{h}^{''}_{jd_j}(x))^T$, since $\tilde{\bm{h}}^{''}_j(x) =  U_j^T\dot{S}_j^{-\frac{1}{2}}\dot{\bm{h}}^{''}_j(x)$, we have
\begin{align*}
    \tilde{S}_j &:= [\int \tilde{h}^{''}_{jk_1}(x)\tilde{h}^{''}_{jk_2}(x)\mathrm{d}x]_{k_1=1,k_2=1}^{d_j,d_j} = \int \tilde{\bm{h}}^{''}_j(x) \tilde{\bm{h}}^{''}_j(x)^T \mathrm{d}x \\
    &= U_j^T\dot{S}_j^{-\frac{1}{2}}\int \dot{\bm{h}}^{''}_j(x) \dot{\bm{h}}^{''}_j(x)^T \mathrm{d}x \dot{S}_j^{-\frac{1}{2}}U_j \\
    &= U_j^T\dot{S}_j^{-\frac{1}{2}}[\int \dot{h}^{''}_{jk_1}(x)\dot{h}^{''}_{jk_2}(x)\mathrm{d}x]_{k_1=1,k_2=1}^{d_j,d_j} \dot{S}_j^{-\frac{1}{2}}U_j\\
    &= U_j^T\dot{S}_j^{-\frac{1}{2}} \dot{S}_j \dot{S}_j^{-\frac{1}{2}}U_j = I_{d_j}
\end{align*}
Thus (iv) proved.
\end{proof}

\section{}
\label{sec:appendix_parameter_reduction}
%%In this part of appendix we prove proposition \ref{prop:reduce_to_univariate_problem}
%%\begin{enumerate}
%%    \item we prove proposition $\ref{prop:optimal_mu_phi}$ in section \ref{sec:one_block_problem} by showing that for one block problem \eqref{eq:one_block_problem}, when $r_j^2$ is fixed, the optimal solution of $\hat{\bm{\mu}}_j$ and $\hat{\Phi}_j$ take the form stated in proposition $\ref{prop:optimal_mu_phi}$ as functions of $r_j^2$.
%%    \item we give the derivation detail of plugin the expression of optimal $\hat{\bm{\mu}}_j$ and $\hat{\Phi}_j$ in proposition $\ref{prop:optimal_mu_phi}$ back to \eqref{eq:one_block_problem} to reduce it to the univariate optimization problem \eqref{eq:univariate_problem} w.r.t $r_j^2$.
%%\end{enumerate}

%%To start, we give the proof of proposition \ref{prop:optimal_mu_phi}. 

In this part of appendix we prove proposition \ref{prop:reduce_to_univariate_problem}. For readers' convenience, we copy the one block problem \eqref{eq:one_block_problem} here:
\begin{align}
    \mathcal{L}_j &= - \E_{\bm{\beta}_j \sim \N(\bm{\mu}_j, \Phi_j)} [\log \N(\mathbf{y}_{(-j)}|Z_j\bm{\beta}_j, \alpha I_n)] +  \KLD{\N(\bm{\mu}_j, \Phi_j)}{\N(\mathbf{0}, r_j^2 I_{d_j})} \label{eq:one_block_problem_appendix}
    \\
    &\overset{c}{=} \frac{1}{\alpha}\{\|\mathbf{y}_{(-j)} - Z_j\bm{\mu}_j\|_2^2 + tr(\Phi_jZ_j^TZ_j)\} + d_j\log r_j^2-\log\det\Phi_j + \frac{\|\bm{\mu}_j\|_2^2 + tr(\Phi_j)}{r_j^2} \nonumber 
    % \label{eq:one_block_problem_spec_appendix}
\end{align}
and copy expression of univariate objective function $G_{\alpha, \bm{\eta}_j, V_j}(r^2)$ in \eqref{eq:univariate_objective} that depends on tuning parameter $\alpha$, j-th gram matrix $V_j := Z_j^TZ_j$ and $\bm{\eta}_j := Z_j^T\mathbf{y}_{(-j)}$ here:
\begin{equation}
\label{eq:univariate_objective_appendix}
    G_{\alpha, \bm{\eta}_j, V_j}(r^2):=\sum_{k=1}^{d_j} \{\alpha \log (v_{jk} r^2 + \alpha) - \frac{\eta_{jk}^2r^2}{v_{jk}r^2 + \alpha}\}.
\end{equation}
Here $v_{jk}$ is the k-th diagonal element of $V_j$ and $\eta_{jk}$ is the k-th element of vector $\bm{\eta}_j$.

Below, we copy Proposition \ref{prop:reduce_to_univariate_problem} here and provide proof in the following:
\begin{prop1}
\label{prop:reduce_to_univariate_problem_appendix}
$(\hat{r}_j^2, \hat{\bm{\mu}}_j, \hat{\Phi}_j)$ minimize objective \eqref{eq:one_block_problem_appendix} if only if
\begin{align}
    \hat{r}_j^2 &= \underset{r^2 \geq 0}{\argmin} \quad G_{\alpha, \bm{\eta}_j, V_j}(r^2) %\label{eq:univariate_problem_appendix}
    \nonumber \\
    \hat{\bm{\mu}}_j &= \hat{r}_j^2 \cdot (\hat{r}_j^2Z_j^TZ_j + \alpha I_{d_j})^{-1}Z_j^T\mathbf{y}_{(-j)} \label{eq:mu_expression_appendix_proof}\\
    \hat{\Phi}_j &= \hat{r}_j^2 \cdot \alpha (\hat{r}_j^2Z^T_jZ_j + \alpha I_{d_j})^{-1}\label{eq:phi_expression_appendix_proof}
\end{align}
\end{prop1}
\begin{proof}
We start with showing that for any fixed $r_j^2 \geq 0$, $\mathcal{L}_j$ in \eqref{eq:one_block_problem_appendix} is marginally optimized at $\hat{\bm{\mu}}$ and $\hat{\Phi}_j$ given in \eqref{eq:mu_expression_appendix_proof} and \eqref{eq:phi_expression_appendix_proof}.

Whenever $r_j^2 = 0$, it's trivial to see that both $\bm{\mu}_j$ and $\Phi_j$ must be exact zero, otherwise the KL divergence part is infinity, thus in this case \eqref{eq:mu_expression_appendix_proof} and \eqref{eq:phi_expression_appendix_proof} are true. 

Whenever $r_j^2 \neq 0$, denote $p(\mathbf{y}_{(-j)}|\bm{\beta}_j):= N(\mathbf{y}_{(-j)}|Z_j\bm{\beta}_j, \alpha I_n)$, and $p(\bm{\beta}_j|\mathbf{y}_{(-j)})$ to be the posterior distribution of $\bm{\beta}_j$ given prior $p(\bm{\beta}_j) = N(\bm{\beta}_j|\mathbf{0}, r_j^2I_{d_j})$, objective \eqref{eq:one_block_problem_appendix} can be written in the following form,
\begin{align*}
    \mathcal{L}_j &\propto \E_{\bm{\beta}_j \sim N(\bm{\mu}_j, \Phi_j)} (\log \frac{N(\bm{\mu}_j, \Phi_j)}{p(\mathbf{y}_{(-j)}|\bm{\beta}_j)p(\bm{\beta}_j)})\\
    &\propto \KLD{N(\bm{\mu}_j, \Phi_j)}{p(\bm{\beta}_j|\mathbf{y}_{(-j)})}\\
    & = \KLD{N(\bm{\mu}_j, \Phi_j)}{N((Z_j^TZ_j + \frac{\alpha}{r^2_j}I_{d_j})^{-1}Z_j^T\mathbf{y}_{(-j)}, (\frac{Z^T_jZ_j}{\alpha} + \frac{I_{d_j}}{r_j^2})^{-1})}
\end{align*}
Since KL Divergence is minimized when two distribution is identical (almost surely), thus when $r_j^2$ is fixed at non-zero value, we have optimal
\begin{align*}
    \hat{\bm{\mu}}_j(r_j^2) &= (Z_j^TZ_j + \frac{\alpha}{r^2_j}I_{d_j})^{-1}Z_j^T\mathbf{y}_{(-j)}\\
    \hat{\Phi}_j(r_j^2) &= \Big (\frac{Z^T_jZ_j}{\alpha} + \frac{I_{d_j}}{r_j^2} \Big )^{-1}
\end{align*}
Rearrange terms, we get expression of \eqref{eq:mu_expression_appendix_proof} and \eqref{eq:phi_expression_appendix_proof}.

Next, we complete the proof by showing that plugging in the expression of optimal $\hat{\bm{\mu}}_j(r_j^2)$ and $\hat{\Phi}_j(r_j^2)$ in \eqref{eq:mu_expression_appendix_proof} and \eqref{eq:phi_expression_appendix_proof} back to $\mathcal{L}_j$ in \eqref{eq:one_block_problem_appendix}, it reduced to a problem
of finding optimal $\hat{r}_j^2$ of the univariate function $G_{\alpha, \bm{\eta}_j, V_j}(r^2_j)$ with expression \eqref{eq:univariate_objective_appendix}.

We start with expression of $\mathcal{L}_j$:
\begin{align}
    \mathcal{L}_j &\overset{c}{=} \frac{1}{\alpha}\{\|\mathbf{y}_{(-j)} - Z_j\bm{\mu}_j\|_2^2 + tr(\Phi_jZ_j^TZ_j)\} + d_j\log r_j^2-\log\det\Phi_j + \frac{\|\bm{\mu}_j\|_2^2 + tr(\Phi_j)}{r_j^2} \nonumber \\
    &\propto \|\mathbf{y}_{(-j)} - Z_j\bm{\mu}_j\|_2^2 + tr(\Phi_jZ_j^TZ_j) + \alpha\cdot\{d_j\log r_j^2-\log\det\Phi_j + \frac{\|\bm{\mu}_j\|_2^2 + tr(\Phi_j)}{r_j^2}\} \nonumber \\
    &\overset{c}{=} -2\mathbf{y}_{(-j)}^TZ_j\bm{\mu}_j +  \bm{\mu}_j^TZ_j^TZ_j\bm{\mu}_j + tr(\Phi_jZ_j^TZ_j) + \alpha\cdot\{d_j\log r_j^2-\log\det\Phi_j + \frac{\|\bm{\mu}_j\|_2^2 + tr(\Phi_j)}{r_j^2}\} \nonumber \\
    &= -2\bm{\eta}_j^T\bm{\mu}_j + \bm{\mu}_j^TV_j\bm{\mu}_j + tr(\Phi_jV_j) + \alpha\cdot\{d_j\log r_j^2-\log\det\Phi_j + \frac{\|\bm{\mu}_j\|_2^2 + tr(\Phi_j)}{r_j^2}\}\label{eq:plugin_derivation_1}
\end{align}
Next we plugin optimal $\hat{\bm{\mu}}(r_j^2)$ and $\hat{\Phi}_j(r_j^2)$ from \eqref{eq:mu_expression_appendix_proof} and \eqref{eq:phi_expression_appendix_proof}. Notice that:
\begin{align}
    \hat{\bm{\mu}}_j(r_j^2) &= r_j^2 \cdot (r_j^2Z_j^TZ_j + \alpha I_{d_j})^{-1}Z_j^T\mathbf{y}_{(-j)} = vector\{\frac{\eta_{jk}\hat{r}_j^2}{v_{jk}\hat{r}_j^2 + \alpha}\}_{k=1}^{d_{j}}\label{eq:mu_expression_spec_appendix}\\
    \hat{\Phi}_j(r_j^2) &= r_j^2 \cdot \alpha (r_j^2Z^T_jZ_j + \alpha I_{d_j})^{-1} =  diag\{\frac{\alpha \hat{r}_j^2}{v_{jk}\hat{r}_j^2 + \alpha}\}_{k=1}^{d_j}\label{eq:phi_expression_spec_appendix}
\end{align}
Plugin \eqref{eq:mu_expression_spec_appendix} and \eqref{eq:phi_expression_spec_appendix} to $\mathcal{L}_j$ expression \eqref{eq:plugin_derivation_1}, we get:
\begin{align*}
    \mathcal{L}_j &\overset{c}{=} \sum_{k=1}^{d_j}\{-2\frac{\eta_{jk}^2r_j^2}{v_{jk}r_j^2 + \alpha} + \frac{\eta_{jk}^2 v_{jk}r^4_j}{(v_{jk}r_j^2 + \alpha)^2} + \frac{\alpha v_{jk}r_j^2}{v_{jk}r_j^2 + \alpha} + \alpha\log(v_{jk}r_j^2 + \alpha) + \frac{\alpha \eta_{jk}^2r_j^2}{(v_{jk}r_j^2 + \alpha)^2} + \frac{\alpha^2}{v_{jk}r_j^2 + \alpha}\}\\
    &=\sum_{k=1}^{d_j}\{\frac{(\alpha v_{jk}^2 - \eta_{jk}^2v_{jk})r_j^4 + (2\alpha^2v_{jk} - \alpha \eta_{jk}^2)r_j^2 + \alpha^3}{(v_{jk}r_j^2 + \alpha)^2} + \alpha\log(v_{jk}r_j^2 + \alpha)\} \\
    &= \sum_{k=1}^{d_j}\{\frac{[(\alpha v_{jk} - \eta_{jk}^2)r_j^2 + \alpha^2]\cdot(v_{jk}r_j^2 + \alpha)}{(v_{jk}r_j^2 + \alpha)^2} + \alpha\log(v_{jk}r_j^2 + \alpha)\}\\
    &= \sum_{k=1}^{d_j}\{\frac{(\alpha v_{jk} - \eta_{jk}^2)r_j^2 + \alpha^2}{v_{jk}r_j^2 + \alpha} + \alpha\log(v_{jk}r_j^2 + \alpha)\}\\
    &\overset{c}{=} G_{\alpha, \bm{\eta}_j, V_j}(r^2_j)
\end{align*}
\end{proof}

\section{}
\label{sec:appendix_univariate_solution}
In this part of the appendix, we provide proof for Proposition \ref{prop:univariate_problem_interval} of our paper to show that the optimal solution of univariate problem \eqref{eq:univariate_problem} must be within closed interval $[l_j, u_j]$. For readers' convenience, we copy the target univariate problem \eqref{eq:univariate_problem} here:
\begin{equation}
\label{eq:univariate_problem_appendix_C}
    \hat{r}_j^2 = \underset{r^2 \geq 0}{\argmin} \quad G_{\alpha, \bm{\eta}_j, V_j}(r^2)
\end{equation}
The univariate function $G_{\alpha, \bm{\eta}_j, V_j}(r^2)$ is a function depends on tuning parameter $\alpha$, j-th gram matrix $V_j = Z_j^TZ_j$ and $\bm{\eta}_j = Z_j^T\mathbf{y}_{(-j)}$. $G_{\alpha, \bm{\eta}_j, V_j}(r^2)$  takes the following form:
\begin{equation}
\label{eq:univariate_objective_appendix_C}
    G_{\alpha, \bm{\eta}_j, V_j}(r^2):=\sum_{k=1}^{d_j} \{\alpha \log (v_{jk} r^2 + \alpha) - \frac{\eta_{jk}^2r^2}{v_{jk}r^2 + \alpha}\}
\end{equation}
Here $v_{jk}$ is the k-th diagonal element of $V_j$ and $\eta_{jk}$ is the k-th element of vector $\bm{\eta}_j$.
We copy the expression of closed interval $[l_j, u_j]$ here:
\begin{align}
    l_j &:= \underset{k = 1, \cdots, d_j}{\min}(\frac{\eta^2_{jk} - \alpha v_{jk}}{v_{jk}^2})_{+} \label{eq:interval_l_appendix}\\
    u_j &:= \underset{k = 1, \cdots, d_j}{\max}(\frac{\eta^2_{jk} - \alpha v_{jk}}{v_{jk}^2})_{+} \label{eq:interval_u_appendix}
\end{align}
Below, we copy proposition \ref{prop:univariate_problem_interval} here and provide proof in the following:
\begin{prop2}
\label{prop:univariate_problem_interval_appendix}
The optimal solution for univariate problem \eqref{eq:univariate_problem} exists and must be within closed interval $[l_j, u_j]$.
\end{prop2}
\begin{proof}
To start we denote:
$$
    g_{\alpha, \eta_{jk}, v_{jk}}(r^2):= \alpha \log (v_{jk} r^2 + \alpha) - \frac{\eta_{jk}^2r^2}{v_{jk}r^2 + \alpha}.
$$
Notice that $G_{\alpha, \bm{\eta}_j, V_j}(r^2):=\sum_{k=1}^{d_j} g_{\alpha, \eta_{jk}, v_{jk}}(r^2)$, taking derivative on function $g_{\alpha, \eta_{jk}, v_{jk}}(r^2)$, we have
$$g_{\alpha, \eta_{jk}, v_{jk}}^{'}(r^2) = \alpha\cdot\frac{v_{jk}^2r^2 - (\eta_{jk}^2 - \alpha v_{jk})}{(v_{jk}r^2+ \alpha)^2 }$$
so $g_{\alpha, \eta_{jk}, v_{jk}}(r^2)$ monotone decrease on $[0, (\frac{\eta_{jk}^2 - \alpha v_{jk}}{v_{jk}^2})_{+}]$ and monotone increase on $[(\frac{\eta_{jk}^2 - \alpha v_{jk}}{v_{jk}^2})_{+}, +\infty)$, which implies that $\sum_{k=1}^{d_j} g_{\alpha, \eta_{jk}, v_{jk}}(r^2)$ monotone decrease on $[0, \underset{k = 1, \cdots, d_j}{\min}(\frac{\eta^2_{jk} - \alpha v_{jk}}{v_{jk}^2})_{+}]$ and monotone increase on $[\underset{k = 1, \cdots, d_j}{\max}(\frac{\eta^2_{jk} - \alpha v_{jk}}{v_{jk}^2})_{+}, +\infty)$. Thus optimal solution must exist and must be within closed interval $[l_j, u_j]$. 
\end{proof}

\section{}
\label{sec:appendix_alpha_implementation}
\subsection{Cross Validation to Choose $\alpha$}
The value of the hyperparameter $\alpha$ plays a crucial role in achieving sparsity in the estimated coefficients. As $\alpha$ increases, the algorithm tends to produce more sparsity in the estimated coefficients. In fact, whenever $\alpha$ surpasses a threshold defined as $\underset{j = 1, \cdots, 2p}{\max}\{{\underset{k = 1, \cdots, d_j}{\max}\frac{(Z_j^T\mathbf{y})_k^2}{v_{jk}}}\}$, all elements of $\hat{\bm{\mu}}_j$ will remain at zero if initialized as such in our algorithm. This indicates that $\alpha$ plays a role similar to the penalty parameter $\lambda$ in the R package \texttt{glmnet} \citep{glmnet}. Larger values of $\alpha$ will include fewer features, while smaller values of $\alpha$ tend to include more features. Keeping that in mind, in order to select the hyperparameter $\alpha$, our methodology borrows ideas used in the R package \texttt{glmnet} \citep{glmnet}, with some adaptations to suit our practical needs. For a given dataset, we predefine a sequence of candidate $\alpha$ values and perform cross-validation on this sequence to identify the optimal $\alpha$.

To speed up computation, we incorporate the warm start concept from \cite{glmnet} to train a sequence of models for the candidate $\alpha$ values. Except for the first $\alpha$, we use the converged values of $\{\hat{\bm{\mu}}_j\}_1^{2p}$  from the previous candidate $\alpha$ as the initialization point for training the model with the current candidate $\alpha$. While \cite{glmnet} suggests training the sequence in descending order of hyperparameters, we find that this approach often leads to suboptimal results for smaller $\alpha$ values in our algorithm. Therefore, we opt for training the sequence in ascending order of $\alpha$, which works nearly as well as training for every $\alpha$ without the warm start.

When plotting the cross-validation mean squared error against the candidate $\alpha$ sequence, we typically observe a U-shaped curve as long as the candidate $\alpha$ values are well-distributed across the range $(0, \underset{j = 1, \cdots, 2p}{\max}\{{\underset{k = 1, \cdots, d_j}{\max}\frac{(Z_j^T\mathbf{y})_k^2}{v_{jk}}}\}]$. Similar to \texttt{glmnet}, we also observe in our algorithm that the $\alpha$ with the smallest cross-validation mean squared error often tends to overfit the data by including many irrelevant features. To address this problem, we borrow the concept of `$\lambda_{1se}$' from the R package \texttt{glmnet} (\cite{glmnet}). In our algorithm, a suitable heuristic value for the hyperparameter $\alpha$ is the maximum $\alpha$ value that falls within $0.15$ times the standard deviation of the minimum cross-validation mean squared error. We denote this $\alpha$ as $\alpha_{0.15se}$.

\subsection{Other Implementation Details}
\begin{description}
\item[Initialization] In coordinate descent algorithms, a proper initialization can lead to improved results. In the context of linear models, \cite{ray2021variational} initialized their algorithm with estimates from ridge regression, while \cite{yang2020variational} used lasso estimates. In our additive model setup, we experimented with ridge and group lasso initializations but did not observe significant improvements in our algorithm's performance. As a practical choice, we simply initialize our algorithm with all $\hat{\bm{\mu}}_j = \mathbf{0}$.
%%\item[Updating Order] Many coordinate descent algorithm are sensitive to updating order of features in every iteration (see \cite{ray2021variational}). We observe similar issue in our algorithm. To resolve it, we apply \cite{ray2021variational}'s strategy. Specifically, we first compute a rudimentary ridge estimate $\hat{\bm{\mu}} = (\bm{X}^T\bm{X} + \underset{j = 1\cdots p}{\min}\|\bm{x}_j\|_2^2\cdot \bm{I}_p)^{-1}\bm{X}^T\mathbf{y}$, then fix the updating order to be the descending order of the magnitude of $\hat{\bm{\mu}}$.
\item[Convergence Rule] Similar to R package \texttt{glmnet}, we halt our algorithm when the difference in residual sum of squares between two iterations is less than $0.1^{6} \cdot \|\mathbf{y} - \bar{\mathbf{y}}\|_2^2$.
\item[Active Set Strategy] o expedite computation, especially in high-dimensional scenarios with sparse true signals, we employ an active set strategy inspired by \cite{glmnet}. This approach involves running one iteration of coordinate descent, focusing on non-zero dimensions until convergence, and then including all features for a final iteration to check for convergence.
\item[Categorical Features] In real-world applications, data often contain a mix of numerical and categorical features. For categorical features, we apply one-hot (dummy) encoding to all levels, center them, and append each level as one $Z_j$ with $d_j = 1$ to our algorithm for coefficient estimation and feature selection.
\end{description}

\section{}
\label{sec:appendix_experiment_setup}
In this section, we list the detailed setup for Experiment 1 and Experiment 2 in Section \ref{sec:performance_comparison}. In the following $\textbf{I}$ is identity matrix, $\textbf{J}$ is all 1 matrix and
$$
    \phi_1(x) = 10e^{-4.6x^2}, \quad
    \phi_2(x) = 4\cos(1.7x), \quad
    \phi_3(x) = 5(x + 1.3)^2, \quad
    \phi_4(x) = 6(x + 5).
$$

\textbf{Experiment 1}:
\begin{itemize}
    \item Case 1: $n = 500$, $p = 10$, $\sigma^2 = 1$, $(f_2, f_5, f_7, f_8) = (\phi_1, \phi_2, \phi_3, \phi_4)$ and all other $f_j$ s are exact zero. Each entry of the design matrix is \textit{i.i.d},  generated from $\mathcal{U}(-1, 1)$ distribution.
    \item Case 2: $n = 800$, $p = 15$, $\sigma^2 = 4$, $(f_1, f_2, f_3, f_4, f_7, f_{11}, f_{12}, f_{13}) = (\phi_1, -\phi_1,\phi_2, -\phi_2,\phi_3, -\phi_3,$ $ \phi_4, -\phi_4)$ and  all other $f_j$ s are exact zero. Each entry of the design matrix is \textit{i.i.d},  generated from $\mathcal{U}(-1, 1)$ distribution.
    \item Case 3: $n = 500$, $p = 150$, $\sigma^2 = 1$, $(f_1, f_4, f_6) = (\phi_2, \phi_3, \phi_4)$ and all other $f_j$ s are exact zero. Each entry of the design matrix is \textit{i.i.d},  generated from $\mathcal{U}(-1, 1)$ distribution.
    \item Case 4: $n = 1000$, $p = 1000$, $\sigma^2 = 1$, $(f_{10}, f_{13}, f_{18}, f_{120}) = (\phi_1, \phi_2, \phi_3, \phi_4)$ and all other $f_j$ s are exact zero. Each entry of the design matrix is \textit{i.i.d},  generated from $\mathcal{U}(-1, 1)$ distribution.
    \item Case 5: $n = 500$, $p = 30$, $\sigma^2 = 1$, $(f_{1}, f_{3}, f_{4}, f_{5}, f_{12}, f_{20}, f_{21}, f_{23}, f_{24}, f_{25}, f_{26}, f_{28}) = (\phi_1, \phi_2, \phi_3,$ $\phi_4, \phi_1, \phi_2, \phi_3, \phi_4, \phi_1, \phi_2, \phi_3, \phi_4)$ and all other $f_j$ s are exact zero. Each row of the design matrix is \textit{i.i.d}, generated from a multivariate normal distribution with mean zero and covariance matrix $\Sigma = 0.7 \textbf{I} + 0.3 \textbf{J}$.
    \item Case 6: Same setup with Case 5 with new covariance matrix $\Sigma = 0.3 \textbf{I} + 0.7\textbf{J}$.
\end{itemize}

\textbf{Experiment 2}:
\begin{itemize}
    \item Case 1: $n = 600$, $p = 18$, $\sigma^2 = 1$, $(f_2, f_3, f_5, f_6, f_8, f_{10})$ are assigned with nonlinear functions $(\phi_1, 2\phi_1, \phi_2, 2\phi_2, \phi_3, 2\phi_3)$, $(f_{11}, f_{12}, f_{14}, f_{15}, f_{17}, f_{18})$ are assigned with linear functions $(\phi_4, 2\phi_4, 3\phi_4, -\phi_4, -2\phi_4, -3\phi_4)$ , and all other $f_j$ s are exact zero. Each entry of the design matrix is \textit{i.i.d},  generated from $\mathcal{U}(-1, 1)$ distribution.
    
    \item Case 2: $n = 2000$, $p = 100$, $\sigma^2 = 1$, $\{f_1, \dots, f_{25} \}$ are assigned with nonlinear function $\phi_1$ with multiplier evenly spaced from $1$ to $5$, $\{f_{26}, \dots, f_{50} \}$ are assigned with linear function $\phi_4$ with multiplier evenly spaced from $1$ to $5$, and $f_j$ that is not listed here are exact zero. Each entry of the design matrix is \textit{i.i.d},  generated from $\mathcal{U}(-1, 1)$ distribution.
    
    \item Case 3: $n = 1000$, $p = 1200$, $\sigma^2 = 1$, $(f_{12},f_{123}, f_{810})$ are assigned with nonlinear functions $(\phi_1, \phi_2, \phi_3)$, $(f_{90}, f_{500}, f_{811})$ are assigned with linear function $-\phi_4$, and all other $f_j$ s are exact zero. Each entry of the design matrix is \textit{i.i.d},  generated from $\mathcal{U}(-1, 1)$ distribution.
    
    \item Case 4: $n = 600$, $p = 30$, $\sigma^2 = 1$, $\{f_{21}, \dots, f_{25} \}$ are assigned with nonlinear function $\phi_3$, $\{f_{26}, \dots, f_{30}\}$ are assigned with nonlinear function $-\phi_3$, $\{f_{11}, \dots,  f_{15}\}$ are assigned with linear function $\phi_4$, $\{f_{16}, \dots, f_{20}$\} are assigned with linear function $-\phi_4$, and all other $f_j$ s are exact zero. Each row of the design matrix is \textit{i.i.d}, generated from a multivariate normal distribution with mean zero and covariance matrix $\Sigma = 0.7 \textbf{I} + 0.3 \textbf{J}$.
    
    \item Case 5: Same setup as Case 4 but with new covariance matrix $\Sigma = 0.3 \textbf{I} + 0.7\textbf{J}$.
    
\end{itemize}

\end{appendices}
\end{document}